\documentclass[12pt]{article}
\usepackage[pagewise]{lineno}
\usepackage{graphics}
\usepackage{amsmath,amssymb, amsbsy, amsthm}
\usepackage{epsfig}
\usepackage{graphicx}
\usepackage{float}
\usepackage[caption=true]{subfig}
\usepackage{multicol}
\usepackage{algcompatible}
\usepackage{enumerate}
\newtheorem{theorem}{Theorem}[section]

\newtheorem{remark}[theorem]{Remark}

\usepackage{cases}
\usepackage{enumerate}
\usepackage[colorlinks=true,urlcolor=blue,citecolor=red]{hyperref}
\usepackage{color,soul}
\usepackage[utf8]{inputenc}
\usepackage{anysize}
\usepackage{geometry}
\usepackage{setspace}
\usepackage{booktabs, caption, fixltx2e}
\usepackage{times}
\usepackage {tikz}
\usetikzlibrary {positioning}
\definecolor {processblue}{cmyk}{0.96,0,0,0}
\usetikzlibrary{arrows}
\usepackage{verbatim}
\usepackage{tabularx}
\usepackage{mathtools}
\DeclarePairedDelimiter\abs{\lvert}{\rvert}

\date{}
\begin{document}
	
	\title{Deciphering dynamics of recent COVID-19 outbreak in India: An age-structured modeling}
	
\author{  Vijay Pal Bajiya$^1$, Jai Prakash Tripathi$^1$
	,Ranjit Kumar Upadhyay$^2$\\
	\\
	\small $^1$Department of Mathematics, Central University of Rajasthan, Kishangarh-305817, Ajmer,
		India\\	 
		\small $^2$	Department of Mathematics \& Computing, Indian Institute of Technology (ISM), Dhanbad-826004, India \vspace{-0.04in}
}
\maketitle
\begin{abstract}
The transmission dynamics of an infectious disease are most sensitive to the social contact patterns in a population of a particular community and to analyze the precautions people use to reduce the transmission of the disease. The social contact pattern depends on the age distribution of the specific community via different location such as  work, school and recreation etc. Therefore, knowing the age-specific prevalence and incidence of the infectious disease is essential for modeling the future burden of the disease and the effectiveness of interventions such as vaccination.  In the present study, we consider an SEIR age-structured multi-group epidemic model to understand the impact of social contact patterns in controlling the disease. 
To observe how fluctuations in social mixing have affected the spreading of the emerging infectious disease, we used synthetic location-specific social contact matrices in the community. For mathematical analysis, we computed the basic reproduction number $(R_0)$ for the system and also illustrated the global behavior of the system in terms of basic reproduction number.
Further, the existence of optimal control for the associated problem has been established and computed mathematically. The transmission rate for the proposed model using the real data of COVID-19 for India from September 1, 2020, to December 31, 2020, has also been estimated.  
We simulated lifting of different non-pharmaceutical interventions by permitting the people to go back for their works in a phased-manners and investigated the effects of returning to work at different stages, accordingly. Our results suggest that awareness of symptomatic infected individuals of age groups $20-49$ years is beneficial to reduce the number of infected individuals when all schools are closed. However, awareness of symptomatic infected individuals of school children age groups also plays a significant role in reducing disease cases when some schools are partially opened. The simulation results also recommend that the number of cases could be reduced in large numbers by controlling the contacts at school and other gathering places.
Interestingly, it has been investigated that the time-dependent transmission rate is more realistic rather than the constant spread rate to COVID-19 for India via estimating transmission rate using the least square method.  Our study suggests that the early and sudden lifting of control measures could lead to other peaks and a high COVID-19 burden, which could be flattened and reduced by relaxing the interventions gradually. We hope that our results would help health policymakers in deciding appropriate and timely age-based vaccination distribution strategies and, therefore, control the disease.
\end{abstract}

\section{Introduction}
Mathematical modeling of infectious disease spreading has become an essential tool for understanding disease dynamics and outbreak patterns. It also provides crucial insights to the policymakers in making timely decisions to control and reduce the burden of the diseases when limited empirical data are available. The accomplishment of mathematical modeling in advising critical decisions to  human shield has been confirmed for various diseases, including pandemic \textit{Influenza, COVID-19, Flu}, etc. \cite {gra2008, gog2020,ves2020,tho2020}.  Since the infection ability of an infected individual to transmission rapidly and disturb many people in a community. Therefore, infectious diseases (for example, \textit{Influenza, COVID-19}, etc.) directly communicated from individual to individual by the respiratory system have been of particular interest for mathematical modeling.  The importance of epidemic models and the usefulness of policies based on these models are dependent on the robustness of the model parameters, which capture different features of the disease \cite{gra2008, due2007}. The key parameter in epidemic modeling is the probability of successful contact between an infectious and a susceptible individual to spread the infection. The accurate contact structures in the epidemic modeling improve the accuracy and efficiency of the prediction and allow us to investigate the effect of control measures targeted at specific locations, such as schools, workplaces, or homes. Therefore, social contact structures of people are the important factors in epidemic modeling to understand the transmission dynamics and investigate the control measures of infectious diseases \cite{pel2009,bec1995}.

Several assumptions are mandatory to make the range of human relations straightforward into controllable mathematical models of communicable diseases transmitted from one person to another person. The essential assumption of a homogeneous mixing population (in which each individual has an equal probability of contact to infected person) has been transformed by the different realistic frameworks, in which the likelihood of contact fluctuates between different groups, which may be defined by the age of individuals. The extent to which individuals specially mix with people of the same age (assortativeness mixing) is a crucial heterogeneity that is now routinely included in epidemic modeling. The heterogeneity has also been made to further represent the fundamental structure of social contact patterns by breakdown the population according to their different locations such as households, schools, workplaces, and other places \cite{pel2009,kie2020}. For the directly transmitted respiratory virus (\textit{COVID-19, Measles, and Influenza}), social mixing patterns influence the risk of individual-level transmission of diseases \cite{war2018, laf2020} and population-level infection dynamics \cite{kuc2014}, as well as the effectiveness of age-specific control measures \cite{bag2013}. The government of India declared a state of emergency and announced the nationwide lockdown and shut down schools. Sports facilities, non-food shops, restaurants, shopping malls, and traveling on people are prohibited from March 25, 2020 (Phase-1 lockdown).  The authority has also instructed to follow the strict distancing measures, avoid unnecessary social interactions, and the local people to stay at home except for essential works. mass gathering was also prohibited, and therefore, social contacts were also reduced significantly. During the lockdown lifting, social activities such as market opening and traveling with social distancing and face-masks were allowed step by step. The effect of lockdown on the progression of COVID-19 in India  has been studied in \cite{bug2020,baj2020}. Thus, social contact patterns essentially play a significant role in understanding the disease progression.

The changes in social mixing patterns among people of various age groups and age distribution affect the prevalence of each age group. The child age-group individuals are subject to make more unnecessary social contacts than adults due to their unawareness  \cite{mos2008} and hence, children may pay more contribution to disease transmission than adults \cite{cau2008, eam2012}. In the early stage of COVID-19 in India, $75\%$ of confirmed cases belong to the actively working population between the age group of 21 to 60 years \cite{moh2020}. The primary reason for this is that the maximum number of persons in those age groups travel internationally for their jobs and businesses. They were also working in essential services during the lockdown period in India.  Therefore, the number of cases or prevalence of the disease depends strongly on the role of different age groups and the mixing patterns. Other age distributions could show significantly different epidemic shapes and the overall impact on the transmission dynamics of infectious diseases.  School closures are considered an essential intervention for an epidemic due to the higher rate of unnecessary social contact in children. Thus, the impact of school closure depends on children's contribution to disease transmission. The early cases of SARS-CoV-2 in Wuhan, China, were concentrated in adults over 40 years of age due to skewed age distribution in China. The assortative mixing patterns between adults could have condensed infection transmission to children in China's very early stages of the COVID-19 outbreak. In many countries (outside China), COVID-19 outbreaks may have been initiated by working-age travelers entering the country \cite{ sal2010} producing a similar excess of adults in the early phases of local epidemics. In both cases, the school closures that happened potentially further declined the transmission among children, but to what degree is unclear.

The heterogeneity of the infected individuals and severity according to their age groups, especially for children and older people, encouraged the interest of several researchers \cite{cao2020,pre2020,baj2021,jos2020}. Some studies have exposed that the severity of infectious diseases increases with the age and morbidity of hospitalized patients \cite{kel2020,zho2020,ver2020}.  Wu et al. \cite{jos2020} have estimated that the global symptomatic case fatality risk of Coronavirus in Wuhan was $1.4\%$, which is significantly lower than the corresponding crude confirmed case fatality risk of $4.5\%.$ Their findings also suggest that the risk of showing symptoms rises by $4\%$ per year in adults who have age groups between $30$ and $60$ years \cite{jos2020}. Davies et al. \cite{dav2020} considered an age-structured mathematical model to examine the epidemic data of COVID-19 from China, Italy, Japan, Singapore, Canada, and South Korea. Their study recommends that there is a robust connection between the age of an individual and the probability of showing symptoms. 

The study of Zou et al. \cite{zou2020} has shown that the viral load in the asymptomatic cases was very similar to that in the symptomatic cases in many situations. Jones et al. \cite{jon2020} established that viral loads in infected children with age group $0-15$ years do not fluctuate significantly from those of older age people. However, older age people are more likely to develop symptoms due to their immunity.

The existing works of the literature suggest that examination of the transmission dynamics between two generations is essential to a better understanding of COVID-19 transmission and most fundamental to studying the mitigation interventions of COVID-19 efficiently. Therefore, social contact patterns among age groups in the community (like work, school, home, and other locations) are a fundamental factor to consider while modeling the transmission of the COVID-19 pandemic. To incorporate these factors, various mathematical models have been established \cite{pre2020,dav2020,ayo2020,chi2020}. Chikina and Pegden \cite{chi2020} studied an age-structured model to explore age-targeted mitigation interventions. They use the numerical values of the parameter from the literature and make discussion using the age-structured temporal series to fit their considered model. However, Davies et al. \cite{dav2020} also illustrate age-related effects in controlling the COVID-19 pandemic. The authors also discuss the efficiency of different control measures by using statistical inference to fit an age-structured SIR model output to empirical data.

In the present study, we use synthetic location-specific contact patterns (contact matrices) for India to observe how these fluctuations of social mixing among different age groups have affected the disease outbreak progression. We modify these in the case of school and workplace closures, reduction in mixing of the population in the general community, and some other intervention scenarios. Using these matrices and the estimated epidemiological parameters for the COVID-19 outbreak in India \cite{dav2020,pre2017,lau2020}, we estimate the infection rate $\beta$ according to lockdown lifting scenarios and simulates the trajectory of COVID-19 outbreak in India through an age-structured (a multi-group) $SEI^aI^sR$ epidemic model for various types of intervention scenarios. The main findings of present work are (i) awareness of the symptomatic infected individuals of age groups $20-49$ years has been an important control measure to COVID-19 in India (ii) the number of cases could be reduced in large number by controlling the contacts at school and other gathering places (iii) the time-dependent transmission rate may be more realistic rather than the constant rate to COVID-19 for India.  

The rest of the paper is formulated as follows. In Section \ref{mod_form}, we formulate an $SEI^aI^sR$ multi-group epidemic model incorporating the social mixing pattern to quantify the control measures of the disease. The computation of basic reproduction number $(R_0)$, local and global stability of equilibria in terms of $R_0$ are discussed in Section \ref{mat_anal}. We also formulated and computed the optimal control problem in Section \ref{opti_1} and \ref{opti_2}. In Section \ref{num_simu}, we presented the numerical simulations to the case study of COVID-19 in India, including the model parameter estimation. We also discussed the effects of awareness of the symptomatic infected individuals and various scenarios of social mixing patterns.  Finally, we discuss the findings and conclude our work in Section \ref{con}.
\section{Model Formulation }\label{mod_form}
In the case of a homogeneous mixing population, We can divide the total population $(N)$ into five main epidemiological classes, susceptible $(S)$, exposed $(E)$, asymptomatic infected $(I^a)$, symptomatic infected $(I^s)$ and recovered $(R)$. Thus, we can achieve the following system of ordinary differential equations to
describe the infectious disease dynamics.

\begin{equation}\label{model}
\begin{aligned}
\frac{dS(t)}{dt}=&\Gamma-\beta S(t) \lambda_1\left(\frac{I^a(t)+(1-\eta) I^s(t)}{N(t)}\right)-\mu S,\\
\frac{dE(t)}{dt}=&\beta S(t) \lambda_1\left(\frac{I^a(t)+(1-\eta) I^s(t)}{N(t)}\right)-(\alpha+\mu)E(t),\\
\frac{dI^a(t)}{dt}=&(1-\sigma)\alpha E(t)-\left(\gamma^a+\delta^a+\mu\right)I^a(t),\\
\frac{dI^s(t)}{dt}=&\sigma\alpha E(t)-\left(\gamma^s+\delta^s+\mu\right)I^s(t),\\
\frac{dR(t)}{dt}=&\gamma^aI^a(t)+\gamma^sI^s(t)-\mu R(t),
\end{aligned}
\end{equation}
where $N(t) = S(t) + E(t) + I^a(t) + I^s(t) + R(t).$ $\Gamma$ represents the new recruitment in the susceptible population and $\mu$ be the natural death rate of individuals. $\alpha$ is the rate at which the exposed are fetching infectious individuals,
and $\sigma$ represents the proportion of symptomatic infectious in total infectious individuals. $\delta^a$ and $\delta^s$  are the mortality rates due to the disease of asymptomatic and symptomatic infectious individuals, respectively. $\gamma^a$ and $\gamma^s$ are the recovery rates of asymptomatic and symptomatic infectious
individuals, respectively. Assuming that the symptomatic infectious individuals avoid unnecessary contact with others (by self-isolation and other precautions), the transmission of disease by the symptomatic infectious individuals is reduced than asymptomatic infectious individuals. Therefore, $\eta$ represents the proportion of contacts avoided by symptomatic infectious individuals due to their awareness. Thus, $\eta$ also represents the awareness of symptomatic infectious individuals. $\lambda_1$ be the average number of contact between the susceptible and infectious individuals, and $\beta$ represents the probability of infection on the contact between infectious and susceptible individuals.\\
During the initial phases of the infectious disease, disease spreading between individuals is statistically independent, which means the probability of making contact between an infectious individual and someone no longer susceptible is very low. in, Epidemiology, the basic reproduction number $R_0$ is used to forecast the curve of an epidemic such that disease will eliminate when $R_0<1$ and persist when $R_0>1$. $R_0$ represents the number of people that an infectious is expected to infect and can be calculated by next-generation matrix approach \cite{van2002}. The basic reproduction number $(R_0)$ for the model system \eqref{model} is given by.
\begin{equation}\label{R01}
R_0=\frac{\beta\lambda_1\alpha(1-\sigma)}{(\alpha+\mu)(\gamma^a+\delta^a+\mu)}+\frac{\beta\lambda_1\alpha(1-\eta)\sigma}{(\alpha+\mu)(\gamma^s+\delta^s+\mu)}
\end{equation}
where $\frac{\beta\lambda_1\alpha(1-\sigma)}{(\alpha+\sigma)(\gamma^a+\delta^a+\mu)}$ and $\frac{\beta\lambda_1\alpha(1-\eta)\sigma}{(\alpha+\mu)(\gamma^s+\delta^s+\mu)}$ represent the contribution of asymptomatic and symptomatic infectious individuals in the disease transmission, respectively.\\
In general, it must be noted that the value of the basic reproduction number is not a biological constant. However,  the basic reproduction number depends on various epidemiological factor like social contact patterns of individuals, applied interventions to control the disease, etc. Also, the number $R_0$ of the disease usually depends on newly susceptible populations. In particular, a value of $R_0 > 1$ specifies that the infectious disease will initiate to spread in the population if there does not exist any intervention. However, a higher value of $R_0$  represents the faster exponential growth of infection in a community. For example, measles is known to be one of the most contagious diseases, with $12\leq  R_0\leq 18$ \cite{fia2017}. The CDC determined that COVID-19 has an $R_0$ approximately value of $5.7$ for the United States \cite{ste2020}. It is near to that of Polio and Rubella \cite{Liu2021}. For many infectious diseases such as COVID-19, the mortality rate due to disease does not distribute homogeneously. However, it varies according to the age groups of individuals \cite{dav2020}.\\
For many diseases, such as COVID-19, the impact of the different age groups on the disease transmission dynamics varies drastically. It may happen due to different contact patterns among age groups and different mixing patterns. Therefore, we formulate the transmission dynamics model for an outbreak of the disease in the heterogeneously mixing population to quantify the impact of the contact patterns and age of individuals. We consider an age-structured epidemic model (multi-group epidemic model) in which each group of a compartment represents the dynamics of individuals of the age groups like $0-04, 05-09, 10-14, 15-19, \cdots, 75-79.$ That means $S_i, E_i, I^a_i, I^s_i$ , and $R_i$ represent the number of susceptible, exposed, asymptomatic, symptomatic, and recovered individuals of $i^{th}$ age group. In this age-structured epidemic model, we also incorporate the contact patterns according to their age groups. It is described by the social contact matrix, which contains the average number of contact between each pair of age groups. Therefore, the force of infection for age-structured epidemic model takes the following form:

\begin{equation*}
\text{force of infection}=\beta S_i(t) \sum_{j=1}^{n}\lambda_{ij}\left(\frac{I^a_j(t)+(1-\eta_j) I^s_j(t)}{N_j(t)}\right)
\end{equation*}
where $N_j = S_j + E_j + I^a_j + I^s_j + R_j$ be the total population of $j^{th}$ group, $\lambda_{ij}$ be the $ij^{th}$ element of the social contact matrix which represents the number of contact between $i^{th}$ and $j^{th}$ age groups, and $\beta$ is the probability of infection of that contact. Here, it is assumed that symptomatic infectious individuals reduce their contacts compared to asymptomatic individuals, i.e., $\eta_j$ is the proportion of contacts avoided by these self-isolating individuals of $j^{th}$ age group (also allowing for compliance rates). It shows the awareness of symptomatic infectious individuals. 
Thus, by incorporating the social contact matrix (contact patterns), the model system \eqref{model} takes the form of following age-structured epidemic model (multi-group epidemic model) for the $i^{th}$ age group.

\begin{equation}\label{mod}
\begin{aligned}
\frac{dS_i(t)}{dt}=&\Gamma_i-\beta S_i(t) \sum_{j=1}^{n}\lambda_{ij}\left(\frac{I^a_j(t)+(1-\eta_j) I^s_j(t)}{N_j(t)}\right)-\mu_i S_i,\\
\frac{dE_i(t)}{dt}=&\beta S_i(t) \sum_{j=1}^{n}\lambda_{ij}\left(\frac{I^a_j(t)+(1-\eta_j) I^s_j(t)}{N_j(t)}\right)-(\alpha_i+\mu_i)E_i(t),\\
\frac{dI^a_i(t)}{dt}=&(1-\sigma)\alpha_iE_i(t)-\left(\gamma^a_i+\delta^a_i+\mu_i\right)I^a_i(t),\\
\frac{dI^s_i(t)}{dt}=&\sigma\alpha_iE_i(t)-\left(\gamma^s_i+\delta^s_i+\mu_i\right)I^s_i(t),\\
\frac{dR_i(t)}{dt}=&\gamma^a_iI^a_i(t)+\gamma^s_iI^s_i(t)-\mu_iR_i(t),
\end{aligned}
\end{equation}
where $i = 1, 2, \cdots, n$ and $n$ represents the number of age-groups  in the proposed model system. The biological meaning of parameters $\Gamma_i,\mu_i,\alpha_i,\sigma,\gamma^a_i,\gamma^s_i,\delta^a_i$ and $\delta^s_i$ for $i^{th}$ age-group are same as described for model \eqref{model}.
\begin{figure}[H]
	\centering
	\includegraphics[width=10cm,height=12cm]{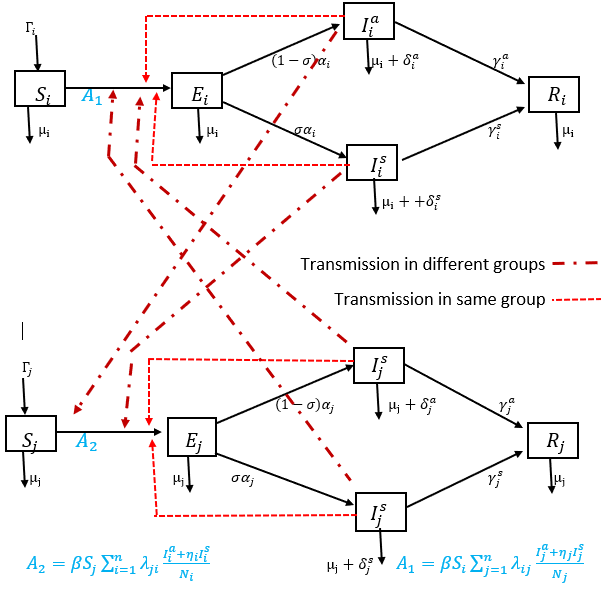}
	\caption{The schematic diagram for proposed model system \eqref{mod} which represents the transmission and transition between different compartments and groups. }
\end{figure}
\begin{remark}
	We can ignore the demographics of the population (i.e., new recruitment in susceptible and natural deaths equal to zero) for the disease for which the infectious period is very shorter compared with the lifespan of an individual of the considered population. For example, in the case of the COVID-19 pandemic, we can ignore the demographics as the infectious period ($7-14$ days) is shorter than the lifespan of an individual (approximately $70$ years). For this purpose, we can put $\Gamma_i=\mu_i=0$ for all $i=1,2,\cdots,n$ in our proposed model system \eqref{mod}.
\end{remark}

\subsection{Characterization of the Social Contact Matrix}
To determine the effect of different interventions, we divide the age-structured contact matrix $\lambda= (\lambda_{ij})_{n\times n}$ into the contributions from home location, schools (all educational institutions like schools, colleges, universities, different types of coaching, etc.), workplaces (govt. and private offices, other workplaces), and other locations (markets, cinema halls, restaurants shopping malls, etc.), given by $\lambda^h, \lambda^s, \lambda^w$ and $\lambda^o$ respectively. The weight of each contact matrix is given by the coefficients $\alpha_h, \alpha_s,\alpha_w$, and $\alpha_o$, we can change the numerical value of these coefficients between zero and one overtime to reflect the effect of different scenarios of interventions. If any of these locations do not contribute to the disease transmission, we set the weight coefficient for the specific contact matrix equal to zero. For a partial contribution, we set the weight coefficient between 0 and 1, accordingly. Therefore, we can divide the contact matrix $\lambda_{ij}$ in four different social contact matrices with corresponding weight coefficients such that
\[\lambda_{ij}=\alpha_h\lambda^h_{ij}+\alpha_s\lambda^s_{ij}+\alpha_w\lambda^w_{ij}+\alpha_o\lambda^o_{ij}\]
For example, schools have been shut down during the strict lockdown in India, so that, we can set $\alpha_s= 0$ for this duration. It must also be distinguished that the contributions to the work locations and other areas like markets are never zero ( $\alpha_w\neq 0$ and  $\alpha_o\neq 0$ ) even during a strict lockdown because people were involved in essential services and markets were opened functioning to a reduced degree. The contributions of contacts at the home location are never zero. Moreover, it could also be noticed that the different types of lockdown can encourage an increase in contact at the home location because people are staying at home more. The contributions of contacts at the home location can not be zero. The selection of weights of social contact patterns also discussed in Moosong et al. \cite{mos2008}.  Data for the number of contacts
made between individuals of each age class are obtained
from an empirical study \cite{pre2017} which estimated contacts separately
for ‘home’, ‘work’, ‘school’ and ‘other’ environments for
5-year age classes up to age 80.

\section{Mathematical Analysis}\label{mat_anal}
In this section, we emphasize the existence, positivity, and boundedness of solutions of model system \eqref{mod}. We also investigate the global dynamics around the equilibria of the system.
\subsection{Nonnegativity and Boundedness}
Due to the biological feasibility of solutions of system \eqref{mod}, we are only interested in nonnegative and bounded solutions of the system. Therefore, we will show that all the solutions of the model system \eqref{mod} with nonnegative initial conditions are nonnegative i.e.,  $S_i(t)\geq  0, E_i(t) \geq 0, I^a_i(t)\geq0, I^s_i(t)\geq 0, R_i(t)\geq 0$ for all $t\geq0$ and $i = 1, 2, \cdots , n.$

\begin{theorem}\label{pos_theo}
The solutions of model system \eqref{mod} with the positive initial condition are nonnegative for all $t\geq0$ and uniformly bounded in the feasible region $\Omega.$
\end{theorem}

\begin{proof}
	First, we show that $S_i(t) > 0$ for all $t> 0.$ For this, we assume that there exists a $t^s_i> 0$ such that, $S_i(t^s_i) = 0$ and $S_i(t) > 0$ for $0 < t < t^s_i.$ Thus, the first equation of system \eqref{mod} becomes
	\begin{equation*}
	\frac{dS_i(t^s_i)}{dt}=\Gamma_i-\beta S_i(t^s_i) \sum_{j=1}^{n}\lambda_{ij}\left(\frac{I^a_j(t^s_i)+(1-\eta_j) I^s_j(t^s_i)}{N_j}\right)-\mu_i S_i(t^s_i)=\Gamma_i>0.
	\end{equation*}
	Therefore, $S_i(t)<0$ for $t\in\left(t^s_i-\epsilon^s,t^s_i\right)$ and sufficiently small positive $\epsilon^s.$ This is a contradiction to $S_i(t) > 0$ for $0 < t < t^s_i.$ Hence, $S_i(t) > 0$ for all $t> 0,$ where $i=1,2,\cdots,n.$ Next, we prove that $E_i(t) > 0$ for $i=1,2,\cdots,n$ and all $t> 0.$ By theory of differential equations, we obtain that the solution of second equation of system \eqref{mod} is given by
	\begin{equation*}
	E_i(t)\geq E_i(0)e^{-(\alpha_i+\mu_i)t}
	\end{equation*}
	Therefore, we have that $E_i(t)\geq 0$ for $t\geq 0.$ In the similiar  manner, it can easily  be proved that $I^a_i(t)\geq 0$, $I^s_i(t)\geq 0$ and $R_i(t)\geq 0$ for $t\geq 0.$\\
	Further, we show the boundedness of solutions of model system \eqref{mod}. From model system \eqref{mod}, we have 
	\[\frac{d}{dt}(S_i+E_i+I_i^a+I_i^s+R_i)=\lambda_i-\mu_i(S_i+E_i+I_i^a+I_i^s+R_i)-\delta^a_i I_i^a-\delta^s_i I_i^s\]
	\[\frac{dN_i}{dt}\leq \Lambda_i-\mu_iN_i \,\,\,\,\text{for}\,\,i^{th} group.\]
	Sinec, we obtain $\lim_{t\to \infty}sup N_i(t) \leq \frac{\Lambda_i}{\mu_i}$. Hence, we have $$\Omega=\left\{\left(S_i,\,E_i,\,I^a_i,\,I^s_i,\,R_i\right)\in \mathbb{R}_+^5|0<S_i+E_i+I^a_i+I^s_i+R_i\leq \frac{\Lambda_i}{\mu_i}\right\}$$ \\
	Thus, the solutions of model system \eqref{mod} is bounded. By the proof procedure of Theorem \ref{pos_theo}, we know that all solutions of system \eqref{mod} ultimately come in and continue in the region $\Omega.$ Hence $\Omega$ is a bounded absorbing set for system \eqref{mod}. Thus, this completes the proof of Theorem \ref{pos_theo}.
\end{proof}

\subsection*{Existence of Solutions}
\begin{theorem}
	The model system \eqref{mod} has a unique solution when initial conditions for all the variables are
	nonnegative i.e., $S_i(0)\geq0, E_i(0) \geq0, I^a_i(0)\geq 0, I^s_i(0)\geq0, R_i(0) \geq 0$ for all $i = 1, 2, \cdots, n.$
\end{theorem}
\begin{proof}
	Let $X=\left(\begin{array}{c}
	X_1\\X_2\\\vdots\\X_n
	\end{array}\right)$ where $X_i=\left(\begin{array}{c}
	S_i(t)\\E_i(t)\\I^a_i(t)\\I^s_i(t)\\R_i(t)
	\end{array}\right)$  for all $i=1,2,\cdots,n.$ So, the model system \eqref{mod} can be written in the following form 	$\Phi(X)=AX+B(X),$ where
	\begin{equation*}A=\left(
	\begin{array}{cccc}
	A_1&0&\cdots&0\\
	0&A_2&\cdots&0\\
	\vdots&\vdots&\vdots&\vdots\\
	0&0&\cdots&A_n
	\end{array}\right)
	\end{equation*}
	with
	\begin{equation*}A_i=\left(
	\begin{array}{ccccc}
	-\mu_i&0&0&0&0\\
	0&-(\alpha_i+\mu_i)&0&0&0\\
	0&0&-(\gamma^a_i+\delta^a_i+\mu_i)&0&0\\
	0&0&0&-(\gamma^s_i+\delta^s_i+\mu_i)&0\\
	0&0&0&0&-\mu_i
	\end{array}\right)
	\end{equation*}
	and  $B(X)=\left(\begin{array}{c}
	B(X_1)\\B(X_2)\\\vdots\\B(X_n)
	\end{array}\right)$ with  $B(X_i)=\left(\begin{array}{c}
	\Gamma_i-\beta S_i(t)\sum_{n}^{j=1}\lambda_{ij}\frac{I^a_j+(1-\eta_j) I^s_j}{N_j}\\\beta S_i(t)\sum_{n}^{j=1}\lambda_{ij}\frac{I^a_j+(1-\eta_j) I^s_j}{N_j}\\0\\0\\0
	\end{array}\right).$
	The functions $B(X_i)$ satisfies
	\begin{equation*}
	\begin{aligned}
	\abs*{B(X^1_i)-B(X^2_i)}=&\abs*{\beta S^1_i(t)\sum_{n}^{j=1}\lambda_{ij}\frac{I^{a1}_j+(1-\eta_j) I^{s1}_j}{N_j}-\beta S^2_i(t)\sum_{n}^{j=1}\lambda_{ij}\frac{I^{a2}_j+(1-\eta_j) I^{s2}_j}{N_j}}\\
	\leq&\abs*{\beta S^1_i(t)\sum_{n}^{j=1}\lambda_{ij}\frac{I^{a1}_j+\eta I^{s1}_j}{N_j}-\beta S^1_i(t)\sum_{n}^{j=1}\lambda_{ij}\frac{I^{a2}_j+(1-\eta_j) I^{s2}_j}{N_j}}\\
	&+\abs*{\beta S^1_i(t)\sum_{n}^{j=1}\lambda_{ij}\frac{I^{a2}_j+(1-\eta_j) I^{s2}_j}{N_j}-\beta S^2_i(t)\sum_{n}^{j=1}\lambda_{ij}\frac{I^{a2}_j+(1-\eta_j) I^{s2}_j}{N_j}}\\
	=&\abs*{S^1_i(t)}\abs*{\beta \sum_{n}^{j=1}\frac{\lambda_{ij}}{N_j}\left(\left(I^{a1}_j-I^{a2}_j\right)+\eta\left(I^{s1}_j-I^{s2}_j\right)\right)}\\&+\abs*{\beta \sum_{n}^{j=1}\lambda_{ij}\frac{I^{a2}_j+(1-\eta_j) I^{s2}_j}{N_j}\left(S^1_i(t)-S^2_i(t)\right)}.\\
	\abs*{B(X^1_i)-B(X^2_i)}	\leq&\frac{\Gamma_i}{\mu_i}\left(\beta \sum_{n}^{j=1}\frac{\lambda_{ij}}{N_j}\left(\abs*{I^{a1}_j-I^{a2}_j}+(1-\eta_j)\abs*{I^{s1}_j-I^{s2}_j}\right)\right)\\&+\beta(2-\eta_j)\sum_{j=1}^{n}\frac{\lambda_{ij}}{N_j}\frac{\Gamma_j}{\mu_j}\abs*{S^1_i(t)-S^2_i(t)}.
	\end{aligned}
	\end{equation*}
	we also have that 
	\begin{equation*}
	\abs*{B(X^1)-B(X^2)}\leq\sum_{i=1}^{n}\abs*{B(X^1_i)-B(X^2_i)}.
	\end{equation*}
	Therefore, we obtain
	\begin{equation*}
	\abs*{B(X^1)-B(X^2)}\leq M\sum_{i=1}^{n}\left(\abs*{I^{a1}_i-I^{a2}_i}+\abs*{I^{s1}_i-I^{s2}_i}+\abs*{S^1_i(t)-S^2_i(t)}\right)
	\end{equation*}
	\[\implies \abs*{B(X^1)-B(X^2)}\leq M\abs{\abs*{X_1-X_2}},\]
	where\[M=\beta(2-\eta_j)\sum_{i=1}^{n}\sum_{j=1}^{n}\frac{\Gamma_i}{\mu_i}\frac{\lambda_{ij}}{N_j}.\]
	Therefore, we have that 
	\[\abs{\abs{\Phi(X_1)-\Phi(X_2)}}\leq M^0\abs{\abs*{X_1-X_2}},\]
	where $M^0=\min\left(M,\abs*{\abs*{A}}\right).$ Therefore, it follows that the function $\Phi$ is uniformly Lipschitz continuous, and the condition on  $S_i(0)\geq0, E_i(0) \geq0, I^a_i(0)\geq 0, I^s_i(0)\geq0, R_i(0) \geq 0$ for all $i = 1, 2, \cdots, n.$ Hence, using results stated in \cite{bir1989} (section 10 of the first chapter), we conclude that a solution of the system \eqref{mod} exists.
\end{proof}

\subsection{Disease Free Equilibrium and Basic Reproduction Number}
The disease free equilibrium (DFE) for model system \eqref{mod} is $$E^0 =\left(S^0_1,0,0,0,0,S^0_2,0,0,0,0,\cdots,S^0_n,0,0,0,0\right)\in \Omega$$
where $S^0_i = \frac{\Gamma_i}{\mu_i}$ and always exists. 
If the value of $R_0$ is less than $1$, then the likelihood of producing new infection cases by infectious people is insufficient for an outbreak to be sustained. If $R_0$ is greater than $1$, the number of secondary cases increases the infection and an epidemic arises until the proportion of susceptible individuals declines. The next-generation matrix method in \cite{van2002} is applied to determine the basic reproduction number $(R_0)$.\\
For calculation of $R_0$, We rewrite the middle $3n$ (second, third and fourth) equations of system \eqref{mod} as follows: $x^{\prime}=\mathcal{F}-\mathcal{V}$ and $x=\left(E_1,E_2,\cdots,E_n,I^a_1,I^a_2,\cdots,I^a_n,I^s_1,I^s_2,\cdots,I^s_n\right)\in\mathbb{R}^{3n},$ where
\begin{equation}\label{math_fv}
\mathcal{F}=\left(\begin{array}{c}
\beta S_1(t) \sum_{j=1}^{n}\lambda_{1j}\left(\frac{I^a_j(t)+(1-\eta_j) I^s_j(t)}{N_j(t)}\right)\\
\beta S_2(t) \sum_{j=1}^{n}\lambda_{2j}\left(\frac{I^a_j(t)+(1-\eta_j) I^s_j(t)}{N_j(t)}\right)\\
\vdots\\
\beta S_n(t) \sum_{j=1}^{n}\lambda_{nj}\left(\frac{I^a_j(t)+(1-\eta_j) I^s_j(t)}{N_j(t)}\right)\\
0\\
0\\
\vdots\\
0\\
0\\
0\\
\vdots\\
0
\end{array}\right)\,\,\text{and}\,\mathcal{V}=\left(\begin{array}{c}
(\alpha_1+\mu_1)E_1(t)\\
(\alpha_2+\mu_2)E_2(t)\\
\vdots\\
(\alpha_2+\mu_2)E_2(t)\\
(1-\sigma)\alpha_1E_1(t)-\left(\gamma^a_1+\delta^a_1+\mu_1\right)I^a_1(t)\\
(1-\sigma)\alpha_2E_2(t)-\left(\gamma^a_2+\delta^a_2+\mu_2\right)I^a_2(t)\\
\vdots\\
(1-\sigma)\alpha_nE_n(t)-\left(\gamma^a_n+\delta^a_n+\mu_n\right)I^a_i(t)\\
\sigma\alpha_1E_1(t)-\left(\gamma^s_1+\delta^s_1+\mu_1\right)I^s_1(t)\\
\sigma\alpha_2E_2(t)-\left(\gamma^s_2+\delta^s_2+\mu_2\right)I^s_2(t)\\
\vdots\\
\sigma\alpha_nE_n(t)-\left(\gamma^s_n+\delta^s_n+\mu_n\right)I^s_n(t)
\end{array}\right).
\end{equation}
Further, by calculating the Jacobian matrices $F$ and $V$ at the DFE $(E^0),$ we obtain
\begin{equation}\label{fv}
F=\left(\begin{array}{ccc}
\bf{O}&F_1&F_2\\
\bf{O}&\bf{O}&\bf{O}\\
\bf{O}&\bf{O}&\bf{O}
\end{array}\right)_{3n \times 3n}\,\,\text{and}\,V=\left(\begin{array}{ccc}
V_{11}&\bf{O}&\bf{O}\\
-V_{21}&V_{22}&\bf{O}\\
-V_{31}&\bf{O}&V_{33}
\end{array}\right)_{3n \times 3n},
\end{equation}
where $\bf{O}$ is the zero matrix of order $n\times n$ having all entries equal to zero  and 
\[F_1=\left(\begin{array}{cccc}
\beta \lambda_{11}&\beta \lambda_{12}&\cdots&\beta \lambda_{1n}\\
\beta \lambda_{21}&\beta \lambda_{22}&\cdots&\beta \lambda_{nn}\\
\vdots            &\vdots             &     & \vdots\\
\beta \lambda_{n1}&\beta \lambda_{n2}&\cdots&\beta \lambda_{nn}
\end{array}\right)_{n\times n},\,\,V_{22}=\left(\begin{array}{cccc}
\gamma^a_1+\delta^a_1+\mu_1&0&\cdots&0\\
0&\gamma^a_2+\delta^a_2+\mu_2&\cdots&0\\
\vdots            &\vdots             &     & \vdots\\
0&0&\cdots&\gamma^a_n+\delta^a_n+\mu_n
\end{array}\right)_{n\times n},\]
\[V_{11}=\left(\begin{array}{cccc}
\alpha_1+\mu_1&0&\cdots&0\\
0&\alpha_2+\mu_2&\cdots&0\\
\vdots            &\vdots             &     & \vdots\\
0&0&\cdots&\alpha_n+\mu_n
\end{array}\right)_{n\times n},\,V_{21}=\left(\begin{array}{cccc}
(1-\sigma)\alpha_1&0&\cdots&0\\
0&(1-\sigma)\alpha_2&\cdots&0\\
\vdots            &\vdots             &     & \vdots\\
0&0&\cdots&(1-\sigma)\alpha_n
\end{array}\right)_{n\times n},\]
\[V_{31}=\left(\begin{array}{cccc}
\sigma\alpha_1&0&\cdots&0\\
0&(\sigma\alpha_2&\cdots&0\\
\vdots            &\vdots             &     & \vdots\\
0&0&\cdots&\sigma\alpha_n
\end{array}\right)_{n\times n},\,\text{and}\,V_{33}=\left(\begin{array}{cccc}
\gamma^s_1+\delta^s_1+\mu_1&0&\cdots&0\\
0&\gamma^s_2+\delta^s_2+\mu_2&\cdots&0\\
\vdots            &\vdots             &     & \vdots\\
0&0&\cdots&\gamma^s_n+\delta^s_n+\mu_n
\end{array}\right)_{n\times n},\]
\[F_2=\left(\begin{array}{cccc}
\beta \eta_1 \lambda_{11}&\beta \eta_2\lambda_{12}&\cdots&\beta \eta_n\lambda_{1n}\\
\beta \eta_1\lambda_{21}&\beta \eta_2 \lambda_{22}&\cdots&\beta\eta_n \lambda_{2n}\\
\vdots            &\vdots             &     & \vdots\\
\beta \eta_1\lambda_{n1}&\beta \eta_2\lambda_{n2}&\cdots&\beta\eta_n \lambda_{nn}
\end{array}\right)_{n\times n},\]
The basic reproduction number $(R_0)$ is given by the following equation;
\begin{equation}\label{R0}
R_0=\rho(FV^{-1})=\rho\left(F_{1}V_{21}V^{-1}_{11}V^{-1}_{22}+ F_2V_{31}V^{-1}_{11}V^{-1}_{33}\right)
\end{equation}
Here matrix $F_{1}V_{21}V^{-1}_{11}V^{-1}_{22}$ gives the contribution from asymptomatic infected individuals and  matrix $ F_2V_{31}V^{-1}_{11}V^{-1}_{33}$ gives the contribution  from symptomatic infected individuals.
\subsection{Global Dynamics when $R_0\leq 1$}
This section shows that the disease could be eliminated from a community  when the basic reproduction number $R_0\leq 1$ and whatever size of the initial outbreak and that the infection persists otherwise.
\begin{theorem}\label{gase0}
	Assume contact matrix $\lambda_{ij}$ is irreducible. If $R_0\leq 1$, then the DFE $(E^0)$ is globally asymptotically stable in $\Omega.$  If $R_0>1$, then $E^0$ is unstable and system \eqref{mod} is uniformly persistent and there exists at
	least one endemic equilibrium (EE).
\end{theorem}
\begin{proof}
	Let $x=\left(E_1,E_2,\cdots,E_n,I^a_1,I^a_2,\cdots,I^a_n,I^s_1,I^s_2,\cdots,I^s_n\right)\in\mathbb{R}^{3n}$ and\\ $y=(S_1,S_2,\cdots,S_n,R_1,R_2,\cdots,R_n)\in \mathbb{R}^{2n}$ be the disease compartment and disease-free compartment vector, respectively. We set $f(x,y)=(F-V)x-\mathcal{F}(x,y)+\mathcal{V}(x,y)$ where vector function $F,\,V$ are defined in Eq. \eqref{fv} and vector function  $\mathcal{F}(x,y),\,\mathcal{V}(x,y)$ are defined in Eq. \eqref{math_fv}. Further, for the disease compartments and disease free compartments, the model system \eqref{mod} can be written as
	\begin{equation}\label{com_mod}
	x^{\prime}=(F-V)x-f(x,y)\,\, \text{and}\,y^{\prime}=g(x,y),
	\end{equation}
	where \[f(x,y)=\left(\begin{array}{c}
	\beta\sum_{j=1}^{n}\lambda_{1j}\left(1-\frac{S_1}{N_j}\right)\left(I^a_j+(1-\eta_j) I^s_j\right)\\
	\beta\sum_{j=1}^{n}\lambda_{2j}\left(1-\frac{S_2}{N_j}\right)\left(I^a_j+(1-\eta_j) I^s_j\right)\\
	\vdots\\
	\beta\sum_{j=1}^{n}\lambda_{nj}\left(1-\frac{S_n}{N_j}\right)\left(I^a_j+(1-\eta_j) I^s_j\right)\\
	0\\
	0\\
	\vdots\\
	0\\
	0\\
	0\\
	\vdots\\
	0
	\end{array}\right)\,\,\,\text{with}\,\,f(0,y)=0.\]
	It can be easily observed that $f(x,y)\geq 0$ in $\Omega\subset \mathbb{R}^{3n+2n}_+\,,F$ and $V$ are nonnegative matrices.  $P_0=(0,y^0)$ be the disease free equilibrium for the system \eqref{com_mod} that is equivalent to DFE $E^0$ of system \eqref{mod}. Since $y^{\prime}=g(0,y)$ has a unique positive equilibrium $y^0=\left(S_1^0,S^0_2,\cdots,S^0_n,0,0,\cdots,0\right)\in \mathbb{R}^{2n},$  then it is globally asymptotically stable in $\mathbb{R}^{2n}.$ Since $\lambda_{ij}$ is irreducible, therefore $V^{-1}F$ is also irreducible and nonnegative. It follows by Perron-Frobenius theory \cite{ber1994} that $c_T$ be the nonnegative left eigenvector of the matrix $V^{-1}F$ corresponding to the eigenvalue $R_0=\rho(V^{-1}F).$ Further, we assume $L_{DFE}=c^TV^{-1}x$ is a  Lyapunov function for the system \eqref{com_mod} on the region $\Omega.$ By differentiating $L_{DFE}$ along the solution of the system \eqref{com_mod}, we obtain
	\begin{equation}
	\begin{aligned}
	L^{\prime}_{DFE}=&c^TV^{-1}x^\prime=c^TV^{-1}\left((F-V)x-f(x,y)\right)=c^TV^{-1}(F-V)x-c^TV^{-1}f(x,y)\\
	=&(R_0-1)c^Tx-c^TV^{-1}f(x,y).
	\end{aligned}
	\end{equation}  
	Thus, if $R_0\leq 0,$ then $L^{\prime}_{DFE}\leq 0$ in $\Omega$ and $L^{\prime}_{DFE}=0$   implies that  $c^Tx=0$ hence $x=0.$ Using the global
	stability for disease free system $y^\prime=g(0,y)$ and $f(0,y)=0,$ singleton $\{P_0\}$ is the only invariant set in $\mathbb{R}^{5n}_+$ where  $L^{\prime}_{DFE}=0.$ When $R_0=1,$ singleton $\{P_0\}$ is largest invariant set where $L^{\prime}_{DFE}=c^TV^{-1}f(x,y)=0.$  Therefore, by LaSalle's invariance principle \cite{las1976}, $P_0$ is globally asymptotically stable in $\Omega$ when $R_0\leq 1 .$\\
	If $R_0>1$ then $L^{\prime}_{DFE}>0$ on condition that $x>0$ and $y=y^0.$ By continuity of $L^{\prime}_{DFE}>0$ in the interval $(P_0-\epsilon,P_0+\epsilon)$ where $\epsilon$ is a small number, one can say that all the solutions of system \eqref{com_mod} in the positive cone and near by $P_0$ diverges from $P_0$, except those on the invariant y-axis. This implies that
	$P_0$ is unstable. Using the uniform persistence result form \cite{fre1994} and an argument given in the proof of Proposition 3.3 of \cite{li1999}, it can be shown that when $R_0 > 1,$ the uniform persistence of system \eqref{mod} is assured by instability of $P_0.$ Further, the existence of at least one EE can be confirmed by using the uniform persistence and the concept of positive invariance of compact set $\Omega.$
\end{proof}

\subsection{Global Dynamics when $R_0> 1$}
This section illustrates that the endemic equilibrium is unique and globally asymptotically stable (GAS) in the interior of the feasible region $\Omega$. Biologically, we say that the disease always remains in the population and persists at a unique endemic level, whether the initial outbreak's size is small or larger. By Theorem \ref{gase0}, we have that an endemic equilibrium $$E^*=\left(\bar{S_1},\bar{S_2},\cdots,\bar{S_n},\bar{E_1},\bar{E_2},\cdots,\bar{E_n},\bar{I^a_1},\bar{I^a_2},\cdots,\bar{I^a_n},\bar{I^s_1},\bar{I^s_2},\cdots,\bar{I^s_n},\bar{R_1},\bar{R_2},\cdots,\bar{R_n}\right)$$
exists and satisfy the following system of equations:
\begin{equation}\label{EE_sys}
\begin{aligned}
\Gamma_i-\beta S_i(t) \sum_{j=1}^{n}\lambda_{ij}\left(\frac{I^a_j(t)+(1-\eta_j) I^s_j(t)}{N_j(t)}\right)-\mu_i S_i=&0,\\
\beta S_i(t) \sum_{j=1}^{n}\lambda_{ij}\left(\frac{I^a_j(t)+(1-\eta_j) I^s_j(t)}{N_j(t)}\right)-(\alpha_i+\mu_i)E_i(t)=&0,\\
(1-\sigma)\alpha_iE_i(t)-\left(\gamma^a_i+\delta^a_i+\mu_i\right)I^a_i(t)=&0,\\
\sigma\alpha_iE_i(t)-\left(\gamma^s_i+\delta^s_i+\mu_i\right)I^s_i(t)=&0,\\
\gamma^a_iI^a_i(t)+\gamma^s_iI^s_i(t)-\mu_iR_i(t)=&0.\\
\end{aligned}
\end{equation}
\begin{theorem}\label{gsee}
	Assume that the contact matrix $\lambda_{ij}$ is irreducible. The system \eqref{mod} has a unique and globally asymptotically stable (GAS) endemic equilibrium $(E^*)$ in $\Omega$ when $R_0>1$
\end{theorem}

\begin{proof}
	To determine the global stability of endemic equilibrium $E^*$ of system \eqref{mod}, let $D_{1i}=S_i-\bar{S_i}-\bar{S_i}\ln\frac{S_i}{\bar{S_i}}+E_i-\bar{E_i}-\bar{E_i}\ln \frac{E_i}{\bar{E_i}},$  $D_{2i}=I^a_i-\bar{I^a_i}-\bar{I^a_i}\ln \frac{I^a_i}{\bar{I^a_i}}+I^s_i-\bar{I^s_i}-\bar{I^s_i}\ln \frac{I^s_i}{\bar{I^s_i}}$ 	and $D_{3i}=\bar{R_i}-\bar{R_i}\ln \frac{R_i}{\bar{R_i}}.$ By differentiating $D_{1i}$ along the solutions of system \eqref{mod}, we obtain
	
	\begin{equation*}
	\begin{aligned}
	D^{\prime}_{1i}=&\left(1-\frac{\bar{S_i}}{S_i}\right)S^{\prime}_i+\left(1-\frac{\bar{E_i}}{E_i}\right)E^{\prime}_i\\
	=&\left(1-\frac{\bar{S_i}}{S_i}\right)\left(\Gamma_i-\beta S_i \sum_{j=1}^{n}\lambda_{ij}\left(\frac{I^a_j+(1-\eta_j) I^s_j}{N_j}\right)-\mu_i S_i\right)\\
	&+\left(1-\frac{\bar{E_i}}{E_i}\right)\left(\beta S_i \sum_{j=1}^{n}\lambda_{ij}\left(\frac{I^a_j+(1-\eta_j) I^s_j}{N_j}\right)-(\alpha_i+\mu_i)E_i\right).\\
	& \text{By using the system of nonlinear equations \eqref{EE_sys}, we obtain}\\
	=&\left(1-\frac{\bar{S_i}}{S_i}\right)\left(\beta \bar{S_i} \sum_{j=1}^{n}\lambda_{ij}\left(\frac{\bar{I^a_j}+(1-\eta_j) \bar{I^s_j}}{N_j}\right)+\mu_i \bar{S_i}-\beta S_i \sum_{j=1}^{n}\lambda_{ij}\left(\frac{I^a_j+(1-\eta_j) I^s_j}{N_j}\right)-\mu_i S_i\right)\\
	&+\left(1-\frac{\bar{E_i}}{E_i}\right)\left(\beta S_i \sum_{j=1}^{n}\lambda_{ij}\left(\frac{I^a_j+(1-\eta_j) I^s_j}{N_j}\right)-\beta \bar{S_i} \sum_{j=1}^{n}\lambda_{ij}\left(\frac{\bar{I^a_j}+(1-\eta_j) \bar{I^s_j}}{N_j}\right)\right)\\
	=&\mu_i\bar{S_i}\left(2-\frac{\bar{S_i}}{S_i}-\frac{S_i}{\bar{S_i}}\right)+2\beta \bar{S_i} \sum_{j=1}^{n}\lambda_{ij}\left(\frac{\bar{I^a_j}+(1-\eta_j) \bar{I^s_j}}{\bar{N_j}}\right)-\beta \bar{S_i} \sum_{j=1}^{n}\lambda_{ij}\left(\frac{\bar{I^a_j}+(1-\eta_j) \bar{I^s_j}}{\bar{N_j}}\right)\frac{\bar{S_i}}{S_i}\\
	&+\beta \bar{S_i} \sum_{j=1}^{n}\lambda_{ij}\left(\frac{I^a_j+(1-\eta_j) I^s_j}{\bar{N_j}}\right)-\beta \bar{S_i} \sum_{j=1}^{n}\lambda_{ij}\left(\frac{\bar{I^a_j}+(1-\eta_j) \bar{I^s_j}}{\bar{N_j}}\right)\frac{E_i}{\bar{E_i}}\\
	&-\beta S_i \sum_{j=1}^{n}\lambda_{ij}\left(\frac{I^a_j+(1-\eta_j) I^s_j}{\bar{N_j}}\right)\frac{\bar{E_i}}{E_i}\\
	\leq &2\beta \bar{S_i} \sum_{j=1}^{n}\lambda_{ij}\left(\bar{I^a_j}+(1-\eta_j) \bar{I^s_j}\right)-\beta \bar{S_i} \sum_{j=1}^{n}\lambda_{ij}\left(I^a_j+(1-\eta_j) I^s_j\right)\frac{\bar{S_i}}{S_i}+\beta \bar{S_i} \sum_{j=1}^{n}\lambda_{ij}\left(I^a_j+(1-\eta_j) I^s_j\right)\\
	&-\beta \bar{S_i} \sum_{j=1}^{n}\lambda_{ij}\left(I^a_j+(1-\eta_j) I^s_j\right)\frac{E_i}{\bar{E_i}}-\beta S_i \sum_{j=1}^{n}\lambda_{ij}\left(I^a_j+(1-\eta_j) I^s_j\right)\frac{\bar{E_i}}{E_i}\\
	=&\sum_{j=1}^{n}\beta\lambda_{ij} \bar{S_i} \bar{I^a_j}\left(2-\frac{\bar{S_i}}{S_i}-\frac{E_i}{\bar{E_i}}+\frac{I^a_j}{\bar{I^a_j}}+\frac{S_iI^a_j\bar{E_i}}{\bar{S_i}\bar{I^a_j}E_i}\right)+\sum_{j=1}^{n}\beta(1-\eta_j)\lambda_{ij} \bar{S_i} \bar{I^s_j}\left(2-\frac{\bar{S_i}}{S_i}-\frac{E_i}{\bar{E_i}}+\frac{I^s_j}{\bar{I^s_j}}+\frac{S_iI^s_j\bar{E_i}}{\bar{S_i}\bar{I^s_j}E_i}\right).\\
	& \text{Using the inequality}\, 1-x\leq \log x, 
	\text{we obtain}\\
	\leq & \sum_{j=1}^{n}\beta\lambda_{ij} \bar{S_i} \bar{I^a_j}\left(\frac{I^a_j}{\bar{I^a_j}}-\frac{E_i}{\bar{E_i}}-\log\frac{I^a_j}{\bar{I^a_j}}+\log \frac{E_i}{\bar{E_i}} \right)+\sum_{j=1}^{n}\beta(1-\eta_j)\lambda_{ij} \bar{S_i} \bar{I^s_j}\left(\frac{I^s_j}{\bar{I^s_j}}-\frac{E_i}{\bar{E_i}}-\log\frac{I^s_j}{\bar{I^s_j}}+\log \frac{E_i}{\bar{E_i}} \right)\\
	=&\sum_{j=1}^{2n}a_{ij} G_{ij} \,\,\, \text{with} \,\,\,\, a_{ij}=\left\{\begin{split}
	&\beta \bar{S_i} \bar{I^a_j} \qquad 1\leq j\leq n,\\
	&  \beta (1-\eta_j) \bar{S_i} \bar{I^s_j}\qquad n+1\leq j\leq 2n.
	\end{split} 
	\right.
	\end{aligned}
	\end{equation*}
	By differentiating $D_{2i}$ along the solutions of system \eqref{mod}, we obtain
	\begin{equation*}
	\begin{aligned}
	D_{n+i}^{\prime}=&\left(1-\frac{\bar{I^a_i}}{I^a_i}\right)(I^a_i)^{\prime}+\left(1-\frac{\bar{I^s_i}}{I^s_i}\right)(I^s_i)^{\prime}\\
	=&\left(1-\frac{\bar{I^a_i}}{I^a_i}\right)\left((1-\sigma)\alpha_iE_i-\left(\gamma^a_i+\delta^a_i+\mu_i\right)I^a_i\right)+\left(1-\frac{\bar{I^s_i}}{I^s_i}\right)\left(\sigma\alpha_iE_i-\left(\gamma^s_i+\delta^s_i+\mu_i\right)I^s_i\right)\\
	=&\left(1-\frac{\bar{I^a_i}}{I^a_i}\right)\left((1-\sigma)\alpha_iE_i-(1-\sigma)\alpha_i\bar{E_i}\frac{I^a_i}{\bar{I^a_i}}\right)+\left(1-\frac{\bar{I^s_i}}{I^s_i}\right)\left(\sigma\alpha_iE_i-\sigma\alpha_i\bar{E_i}\frac{I^s_i}{\bar{I^s_i}}\right)\\
	=&(1-\sigma)\alpha_i\bar{E_i}\left(1-\frac{\bar{I^a_i}}{I^a_i}\right)\left(\frac{E_i}{\bar{E_i}}-\frac{I^a_i}{\bar{I^a_i}}\right)+\sigma\alpha_i\bar{E_i}\left(1-\frac{\bar{I^s_i}}{I^s_i}\right)\left(\frac{E_i}{\bar{E_i}}-\frac{I^s_i}{\bar{I^s_i}}\right)\\
	\leq& (1-\sigma)\alpha_i\bar{E_i}\left(\frac{E_i}{\bar{E_i}}-\frac{I^a_i}{\bar{I^a_i}}-\log \frac{E_i}{\bar{E_i}}+\log\frac{I^a_i}{\bar{I^a_i}}\right)+\sigma\alpha_i\bar{E_i}\left(\frac{E_i}{\bar{E_i}}-\frac{I^s_i}{\bar{I^s_i}}-\log\frac{E_i}{\bar{E_i}}+\log\frac{I^s_i}{\bar{I^s_i}}\right)\\
	=&a_{n+i,i} G_{n+i,i} \,\,\, \text{with} \,\,\,\, a_{n+i,i}=\left\{\begin{split}
	&(1-\sigma)\alpha_i\bar{E_i} \qquad &1\leq i\leq n,\\
	&  \sigma\alpha_i\bar{E_i}\qquad &n+1\leq i\leq 2n.
	\end{split} 
	\right.
	\end{aligned}
	\end{equation*} 	
	
	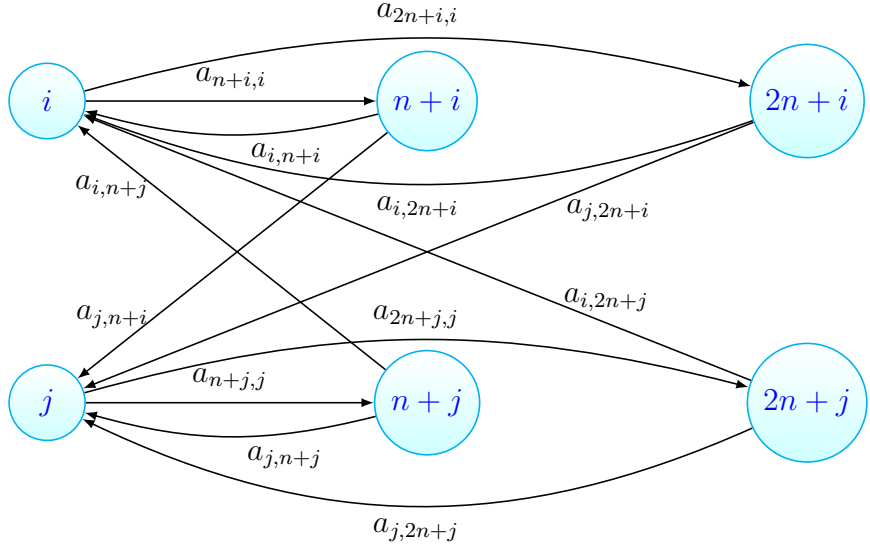
\begin{figure}[h]
		\begin {center}
		\begin {tikzpicture}[-latex ,auto ,node distance =4 cm and 5cm ,on grid ,
		semithick ,
		state/.style ={ circle ,top color =white , bottom color = processblue!20 ,
			draw, processblue , text=blue , minimum width =1 cm}]
		\node[state] (B)
		{$n+i$};
		\node[state] (A) [ left=of B] {$i$};
		\node[state] (C) [right =of B] {$2n+i$};
		\node[state] (D) [below =of A] {$j$};
		\node[state] (E) [below =of B] {$n+j$};
		\node[state] (F) [below =of C] {$2n+j$};

		\path (C) edge [bend left =20] node[below =0.01 cm] {$a_{i,2n+i}$} (A);
		\path (A) edge [bend right = -15] node[above =0.01 cm] {$a_{2n+i,i}$} (C);
		
		\path (F) edge [bend left =25] node[below =0.01 cm] {$a_{j,2n+j}$} (D);
		\path (D) edge [bend right = -15] node[above =0.01 cm] {$a_{2n+j,j}$} (F);
		
		\path (A) edge [bend left =0] node[above] {$a_{n+i,i}$} (B);
		\path (B) edge [bend left =15] node[pos=0.3, below =0.001 cm] {$a_{i,n+i}$} (A);

		\path (D) edge [bend left =0] node[above =0.01 cm] {$a_{n+j,j}$} (E);
		\path (E) edge [bend left =15] node[pos=0.3, below =0.01 cm] {$a_{j,n+j}$} (D);

		\path (E) edge [bend left =0] node[pos=0.9, below =0.15 cm] {$a_{i,n+j}$} (A);
		\path (B) edge [bend left =0] node[pos=0.9, above =0.15 cm] {$a_{j,n+i}$} (D);
		
		\path (F) edge [bend left =0] node[pos=0.2, above =0.01 cm] {$a_{i,2n+j}$} (A);
		\path (C) edge [bend left =0] node[pos=0.2, below =0.10 cm] {$a_{j,2n+i}$} (D);
	\end{tikzpicture}
\end{center}
\caption{The weighted digraph $(\mathcal{G}, \mathcal{A})$ constructed for the model system \eqref{mod} with two group}		
\end{figure}
Hence assumption (1) of Theorem 3.5 stated in \cite{shu2013} holds. To verify assumption (2) of the theorem 3.5 stated in \cite{shu2013}, we define the weighted digraph $(\mathcal{G}, \mathcal{A})$associated with the weight matrix $A=\left[a_{ij}\right]$ with $ a_{ij}>0$ as defined above and all other $a_{ij}=0.$ In digraph $(\mathcal{G}, \mathcal{A})$, there are two kinds of cycles involving direct transmission and cycles involving indirect transmission. For each cycle, assumption (2) of Theorem 3.5 can be verified. Therefore, by Theorem 3.5 discussed in  \cite{shu2013}, we have that $D=\sum_{i=1}^{n}c_iD_i$ is a Lyapunov function for model system \eqref{mod}. Since, $d^{-}(n+i)=d^{-}(2n+i)=1$ for each $i$ by Theorem 3.3 and 3.4 stated in \cite{shu2013}, we have $c_{n+i}=\sum_{j=1}^{n}c_j\frac{a_{j,n+i}}{a_{n+i,i}}.$ Thus, we obtain  \\

\[D=\sum_{i=1}^{n}c_iD_i+\sum_{i=1}^{n}c_{n+i}D_{n+i}=\sum_{i=1}^{n}c_iD_i+\sum_{i=1}^{n}\sum_{j=1}^{n}c_j\frac{a_{j,n+i}}{a_{n+i,i}}D_{n+i}.\]



\[D^\prime\leq\sum_{i=1}^{n}\sum_{j=1}^{n}c_ia_{ij}G_{ij}+\sum_{i=1}^{n}\sum_{j=1}^{n}c_ia_{i,n+j}\left(G_{i,n+j}+G_{n+j,j}\right)+\sum_{i=1}^{n}\sum_{j=1}^{n}c_j\frac{a_{j,n+i}}{a_{n+i,i}}a_{2n+i,n+i}G_{2n+i,n+i}.\]

\[D^\prime\leq\sum_{i=1}^{n}\sum_{j=1}^{n}c_ia_{ij}G_{ij}+\sum_{i=1}^{n}\sum_{j=1}^{n}c_ia_{i,n+j}\left(G_{i,n+j}+G_{n+j,j}\right)+\sum_{i=1}^{n}\sum_{j=1}^{n}c_i\frac{a_{i,n+j}}{a_{n+j,j}}a_{2n+j,n+j}G_{2n+j,n+j}.\]
Let $a^\prime_{ij}=\max\left(a_{ij},\frac{a_{i,n+j}}{a_{n+j,j}}(a_{2n+j,n+j})\right)$ and we obtain

\[D^\prime\leq\sum_{i=1}^{n}\sum_{j=1}^{n}c_ia^\prime_{ij}\left(G_{ij}+G_{2n+j,n+j}\right)+\sum_{i=1}^{n}\sum_{j=1}^{n}c_ia_{i,n+j}\left(G_{i,n+j}+G_{n+j,j}\right).\]

%
\noindent Let $\tilde{a}_{ij}=\max\left(a^\prime_{ij},a_{i,n+j}\right).$ Therefore, we have 
\[D^\prime\leq\sum_{i=1}^{n}\sum_{j=1}^{n}c_i\tilde{a}_{ij}\left(G_{ij}+G_{2n+j,n+j}+G_{i,n+j}+G_{n+j,j}\right).\]
Since, $G_{ij}+G_{2n+j,n+j}+G_{i,n+j}+G_{n+j,j}=2\left(\frac{E_j}{\bar{E_j}}-\frac{E_i}{\bar{E_i}}+\log \frac{E_i}{\bar{E_i}} -\log \frac{E_j}{\bar{E_j}}\right),$ it follows that  
\begin{equation}\label{egeee}
D^\prime\leq2\sum_{i=1}^{n}\sum_{j=1}^{n}c_i\tilde{a}_{ij}\left(\frac{E_j}{\bar{E_j}}-\frac{E_i}{\bar{E_i}}+\log \frac{E_i}{\bar{E_i}} -\log \frac{E_j}{\bar{E_j}}\right)=0.
\end{equation}

\noindent Here, the equality of equation \eqref{egeee} follows from Theorem 3.3 stated in \cite{shu2013}. We can easily verify that   $\{E^*\}$ is the largest invariant set for which $D^\prime=0$. Using this Lyapunov function and LaSalle's invariance principle \cite{las1976}, it follows that $E^*$ is globally asymptotically stable in $\Omega.$  
\end{proof}

\section{The Optimal Control Problem Formulation}\label{opti_1}
There were no treatments and vaccines for COVID-19  available in markets until the end of 2020. As a result, many scientists emphasize mainly three control strategies for challenging this pandemic and reducing the transmission risk with Coronavirus. First, avoid exposure to COVID-19 by the following prevention: reducing or avoiding unnecessary contact with infectious individuals, using face masks,  sensitizers, washing hands with water and soap frequently, cleaning and disinfecting contagion surfaces using proper chemicals. Second, if any susceptible individual make a contact with the infected person directly or indirectly, then encourage them to join the quarantine for provide some special care and protection for the disease. Third, putting infected individuals in isolation or making hospitalization to decrease the infection risk of COVID-19.

Our main objective of this optimal control problem is to minimize the total number of exposed people $(E_i)$, the asymptomatic infectious people $I^a_i$ and the symptomatic infectious people $I^s_i$. For this, we include the three types of control in the controlled system \eqref{con_mod}. The first control $\omega_i$ characterizes the effort of awareness programs in $i^{th}$ age group that target presenting the control measures to prevent individuals from being infected (washing hands regularly, using face masks, avoiding unnecessary contacts with infected surfaces, and individuals). Therefore, the term $(1-\omega_i)$ is used to reduce the force of infection by a proportion. The second control parameter $\theta_i$ represents the effort to encourage the exposed individuals of $i^{th}$ age group to join quarantine centers. Finally, control parameter $\vartheta$ measures the effort to aware the infected individuals to stay home in isolation or join the institutional isolation ward. So, the mathematical system with control parameters  is given by the following system of differential equations.

\begin{equation}\label{con_mod}
\begin{aligned}
\frac{dS_i(t)}{dt}=&\Gamma_i-\beta(1-\omega_i(t)) S_i(t) \sum_{j=1}^{n}\lambda_{ij}\left(\frac{I^a_j(t)+(1-\eta_j) I^s_j(t)}{N_j(t)}\right)-\mu_i S_i,\\
\frac{dE_i(t)}{dt}=&\beta S_i(t)(1-\omega_i(t)) \sum_{j=1}^{n}\lambda_{ij}\left(\frac{I^a_j(t)+(1-\eta_j) I^s_j(t)}{N_j(t)}\right)-(\alpha_i+\mu_i)E_i(t)-\theta_i(t)E_i(t),\\
\frac{dI^a_i(t)}{dt}=&(1-\sigma)\alpha_iE_i(t)-\left(\gamma^a_i+\delta^a_i+\mu_i\right)I^a_i(t)-\vartheta(t)I^a_i(t),\\
\frac{dI^s_i(t)}{dt}=&\sigma\alpha_iE_i(t)-\left(\gamma^s_i+\delta^s_i+\mu_i\right)I^s_i(t)-\vartheta(t)I^s_i(t),\\
\frac{dR_i(t)}{dt}=&\gamma^a_iI^a_i(t)+\gamma^s_iI^s_i(t)-\mu_iR_i(t),
\end{aligned}
\end{equation}
where $S_i(0)\geq 0, E_i (0) \geq 0, I^a_i (0) \geq 0, I^s_i (0) \geq 0,$ and $ R_i (0) \geq 0$ for all $i = 1, 2,\cdots , n$ are given initial sizes of
variables.
The main objective of the optimal control problem is to compare the costs of these intervention strategies and their effectiveness in the fight against the disease. To do this, we need to investigate the optimal level of effort that would be required to control the disease or to minimize the total number of infected individuals. For this, we use the following objective function :
\begin{equation}\label{opt_fun}
\begin{aligned}
J(\omega_i,\theta_i,\vartheta_i)=&\int_{0}^{T}\left(E_i(t)+I^a_i(t)+I^s(t)+\frac{A^1_i}{2}\left(\omega_i(t)\right)^2+\frac{A^2_i}{2}\left(\theta_i(t)\right)^2+\frac{A^3_i}{2}\left(\vartheta_i(t)\right)^2\right)dt,
\end{aligned}
\end{equation}
where  parameters $A^1_i>0, A^2_i>0,
$ and $A^3_i>0,$ for all $i = 1, 2,\cdots, n$ are the cost coefficients at time $t$ associated with applied controls , $T$ is the final
time. On the other hands, we try to find the optimal controls $\omega^*_i,\,\theta^*_i,$, and $\vartheta^*_i$ such that
\begin{equation*}
J(\omega^*_i,\theta^*_i,\vartheta^*_i)=\min_{(\omega_i,\theta_i,\vartheta_i)\in C_{ad}}J(\omega_i,\theta_i,\vartheta_i)
\end{equation*}
where $C_{ad}$ is the set of acceptable control defined by \[C_{ad}=\left\{(\omega_i,\theta_i,\vartheta_i):0\leq \omega_i(t)\leq 1; 0\leq \theta_i(t)\leq 1;0\leq \vartheta_i(t)\leq 1, t\in[0,T]\right\}\]

\subsection*{Existence of an optimal control}\label{opti_2}
\begin{theorem}
	Subject to the control system \eqref{con_mod} with initial conditions, there exists  optimal controls $\omega^*_i,\,\theta^*_i,$, and $\vartheta^*_i$ such that
	\[J(\omega^*_i,\theta^*_i,\vartheta^*_i)=\min_{(\omega_i,\theta_i,\vartheta_i)\in
		C_{ad}}J(\omega_i,\theta_i,\vartheta_i),\]
	if  the following conditions are satisfied:
	\begin{enumerate}[(1)]
		\item The set of controls and corresponding state variables are nonempty.
		\item The control set $C_{ad}$ is compact (convex and closed).
		\item The right-hand side of the state system \eqref{con_mod} is bounded by a linear function in the state and
		control variables.
		\item The integrand $L\left(S_i, E_i, I^a_i, I^s_i,R_i\right) = E_i(t)+I^a_i(t)+I^s(t)+\frac{A^1_i}{2}\left(\omega_i(t)\right)^2+\frac{A^2_i}{2}\left(\theta_i(t)\right)^2+\frac{A^3_i}{2}\left(\vartheta_i(t)\right)^2$ of the
		objective functional is convex on $C_{ad}$ and there exist constants $c_1$ and $c_2$ such that
		\[L\left(S_i, E_i, I^a_i, I^s_i,R_i\right)\geq -3c_1+\frac{c_2}{2}\left(\abs*{\omega_i}^2+\abs*{\theta_i}^2+\abs*{\vartheta_i}^2\right)\]
	\end{enumerate}
\end{theorem}
\begin{proof}
	By results discussed in \cite{fle1975}, the existence of the optimal control can  easily be obtained and its characteristic can be discussed.  First, we shall prove that the set of controls and  corresponding state variables are nonempty by using the simplified form of the existence result in \cite{dip2009}.\\
	Let $Y_1 = S_i, Y_2 = E_i,  Y_3 = I^a_i, Y_4 = I^s_i$ and $Y_5 = R_i $ then from system \eqref{con_mod}, we have that $Y^{\prime}_k =
	F_{Y_k}(t, Y_1, Y_2, Y_3, Y_4, Y_5)$ with $k = 1, 2, 3, 4, 5$ where $F_{Y_k}(t, Y_1, Y_2, Y_3, Y_4, Y_5)$ can be obtained from the right
	hand side of equations of system \eqref{con_mod}. Further, let $\omega, \theta$ and $\vartheta$ are constant control parameters. Since all the model parameters are constants and $Y_1, Y_2, Y_3,  Y_4$ and $Y_5$ are continuous. Therefore $F_{S_i}, F_{E_i}, F_{I^a_i}, F_{I^s_i}$ and $F_{R_i}$ are also continuous. Moreover, the partial derivatives $\frac{\partial F_{Y_k}}{\partial Y_k}$
	are also continuous for	all $k = 1, 2, 3, 4, 5.$ Therefore, a unique solution $(S_i, E_i, I^a_i, I^s_i, R_i)$ is  satisfying the given initial conditions for the variables.
	Hence, there exists a  set of controls and the corresponding state variables. Thus, the first condition is satisfied.

We have that all right-hand sides terms of system \eqref{con_mod} are continuous and also bounded above by a sum of bounded controls and state variables. It also can be written as a linear function of $\omega,\theta$ and $\vartheta$ having the time depend coefficients and state variables.
	From the system \eqref{con_mod}, we obtain
	\[\frac{dN_i(t)}{dt}\leq \Gamma_i-\mu_i N_i\implies \lim \sup_{t\rightarrow \infty}N_i(t)\leq \frac{\Gamma_i}{\mu_i}. \]
	Since,  all solutions for the controlled system \eqref{con_mod} are bounded.  Therefore, there exist positive constants $Z_1,Z_2,Z_3,Z_4,Z_5$ such that for all $t\in[0,T]:$
	$S_i(t)\leq Z_1,E_i(t)\leq Z_2,I^a_i(t)\leq Z_3, I^s_i(t)\leq Z_4,$ and $R_i(t)\leq Z_5.$\\
	Further, we consider
	\begin{equation*}
	\begin{aligned}
	F_{S_i}\leq&\Gamma_i+\omega_i(t) S_i(t) \sum_{j=1}^{n}\lambda_{ij}\left(\frac{I^a_j(t)+(1-\eta_j) I^s_j(t)}{N_j(t)}\right),\\
	F_{E_i}\leq&\beta S_i(t)(2+\eta_j)\sum_{j=1}^{n}\lambda_{ij}-\beta S_i(t)\omega_i(t) \sum_{j=1}^{n}\lambda_{ij}\left(\frac{I^a_j(t)+(1-\eta_j) I^s_j(t)}{N_j(t)}\right)-\theta_i(t)E_i(t),\\
	F_{I^a_i}\leq&(1-\sigma)\alpha_iE_i(t)-\vartheta(t)I^a_i(t),\\
	F_{I^s_i}\leq&\sigma\alpha_iE_i(t)-\vartheta(t)I^s_i(t),\\
	F_{R_i}\leq&\gamma^a_iI^a_i(t)+\gamma^s_iI^s_i(t).
	\end{aligned}
	\end{equation*}
	Therefore, we can rewrite the control system \eqref{con_mod} in the compact form as:
	\begin{equation*}
	F(S_i, E_i, I^a_i, I^s_i, R_i)\leq \Lambda_i+A^1_iX^1_i(t)-B^1_iU_i(t),
	\end{equation*}
	where $	F(S_i, E_i, I^a_i, I^s_i, R_i)=\left(\begin{array}{c}
	F_{S_i}\\
	F_{E_i}\\
	F_{I^a_i}\\
	F_{I^s_i}\\
	F_{R_i}
	\end{array}\right),	\Lambda_i=\left(\begin{array}{c}
	\Gamma_i\\
	0\\
	0\\
	0\\
	0
	\end{array}\right),X_i(t)=\left(\begin{array}{c}
	S_i\\
	E_i\\
	I^a_i\\
	I^s_i\\
	R_i
	\end{array}\right),$  $U_i(t)=\left(\begin{array}{c}
	\omega_i\\
	\theta_i\\
	\vartheta_i
	\end{array}\right)$
	$A^1_i=\left(\begin{array}{ccccc}
	0&0&0&0&0\\
	\beta(1+\eta)\sum_{j=1}^{n}\lambda_{ij}&0&0&0&0\\
	0&(1-\sigma)\alpha_i&0&0&0\\
	0&\sigma\alpha_i&0&0&0\\
	0&0&\gamma^a_i&\gamma^s_i&0
	\end{array}\right)$
	and $B_i=\left(\begin{array}{ccc}
	- S_i(t) \sum_{j=1}^{n}\lambda_{ij}\left(\frac{I^a_j(t)+\eta I^s_j(t)}{N_j(t)}\right)&0&0\\
	S_i(t) \sum_{j=1}^{n}\lambda_{ij}\left(\frac{I^a_j(t)+\eta I^s_j(t)}{N_j(t)}\right)&E_i(t)&0\\
	0&0&I^a_i(t)\\
	0&0&I^s_i(t)\\
	0&0&0\\
	\end{array}\right).$
	It gives a linear function of control and state variable vectors. Since $\abs*{\abs*{\Lambda_i}}\leq \Gamma_i,$ therefore,  we can obtain 
	\begin{equation}
	\begin{aligned}
	\abs*{\abs*{F(S_i, E_i, I^a_i, I^s_i, R_i)}}\leq& \abs*{\abs*{\Lambda_i}}+\abs*{\abs*{A_i}}\abs*{\abs*{X_i(t)}}+\abs*{\abs*{B_i}}\abs*{\abs*{U_i(t)}}\\
	&\Gamma_i+M^{\prime}_i\left(\abs*{\abs*{X_i(t)}}+\abs*{\abs*{U_i(t)}}\right)
	\end{aligned}
	\end{equation}
	where $M^{\prime}_i=\max\left(\abs*{\abs*{A_i}},\abs*{\abs*{B_i}}\right).$ Thus, we have that the right-hand side is bounded by a sum of state and control vectors. Hence, third condition is satisfied.\\
	It is obvious that the objective function \eqref{opt_fun} is convex in $C_{ad}.$ Further, we have that the state variables are bounded. So, let  $c_1=\sup_{t\in[0,T]}\left(E_i(t),I^a_i(t),I^s_i(t)\right)$ and let there exist a constant $c_2$ such that $c_2=\inf\left(A^1_i,A^2_i,A^3_i\right).$ Therefore, we obtain $L\left(S_i, E_i, I^a_i, I^s_i,R_i\right)\geq-3c_1+\frac{c_2}{2}\left(\abs*{\omega_i}^2+\abs*{\theta_i}^2+\abs*{\vartheta_i}^2\right).$ Thus, the forth condition also hold. Hence,  from results stated in \cite{fle1975} (Theorem 2.1 of Chapter 3), we conclude that an optimal control exit for the controlled system \eqref{con_mod}. 
\end{proof}

\subsection*{Characterization of the optimal control}
Here, we compute the necessary condition for the optimal control by using the pontryagin's maximum principle \cite{pon1962}. The idea of Pontryagin’s maximum principle makes known to the adjoint function to the associated system to the objective functional resulting in the formulation of a function called the Hamiltonian. The principle interprets into a problem of minimizing Hamiltonian  $H_i(t)$ at time $t$ defined by

\begin{equation}\label{hami}
H_i(t)=E_i(t)+I^a_i(t)+I^s(t)+\frac{A^1_i}{2}\left(\omega_i(t)\right)^2+\frac{A^2_i}{2}\left(\theta_i(t)\right)^2+\frac{A^3_i}{2}\left(\vartheta_i(t)\right)^2+\sum_{k=1}^{5}\psi^k_i(t)f_k\left(S_i, E_i, I^a_i, I^s_i,R_i\right),
\end{equation}
where $\psi^k_i$ is the $k^{th}$ adjoint variable at time $t$ and $f_k$ is the right-hand side function of the control system  \eqref{con_mod} corresponding to the $k^{th}$ state at  time $t.$
\begin{theorem}\label{the_opti}
Given the optimal controls $\omega^*_i,\theta^*_i, \vartheta^*_i$ and the solutions $S_i, E_i, I^a_i, I^s_i$ and $R_i$ associated to control system \eqref{con_mod}, then there exist adjoint variables $\psi^k_i$ for $k=1,2,3,4,5$ satisfying
\begin{equation}
\begin{aligned}
(\psi^1_i)^\prime==&\psi^1_i\mu_i+(\psi^1_i-\psi^2_i)\beta\left(1-\omega_i(t)\right)\sum_{j=1}^{n}\lambda_{ij}\left(\frac{I^a_j(t)+\eta I^s_j(t)}{N_j}\right),\\
(\psi^2_i)^\prime==&\left(\alpha_i+\mu_i+\theta_i(t)\right)\psi^2_i-(1-\sigma)\alpha_i\psi^3_i-\sigma\alpha_i\psi^4_i,\\
(\psi^3_i)^\prime==&(\psi^1_i-\psi^2_i)\beta(1-\omega_i(t))S_i(t)\sum_{j=1}^{n}\left(\frac{\lambda_{ij}}{N_j}\right)+\left(\gamma^a_i+\delta^a_i+\mu_i+\vartheta(t)\right)\psi^3_i+\gamma^a_i\psi^5_i,\\
(\psi^4_i)^\prime==&(\psi^1_i-\psi^2_i)\beta(1-\omega_i(t))S_i(t)(1-\eta_j)\sum_{j=1}^{n}\left(\frac{\lambda_{ij}}{N_j}\right)+\left(\gamma^s_i+\delta^s_i+\mu_i+\vartheta(t)\right)\psi^4_i+\gamma^s_i\psi^5_i,\\
(\psi^5_i)^\prime==&\mu_i\psi^5_i,
\end{aligned}
\end{equation}	
with the tranversality conditions at time $T:\psi^1_i(T)=0,\psi^2_i(T)=1,\psi^3_i(T)=1,\psi^4_i(T)=1,\psi^5_i(T)=0.$ Then the optimal controls  $\omega^*_i,\theta^*_i,$ and $ \vartheta^*_i$ are given by \[\omega^*_i=\max \left(0,\min \left( 1,\frac{\left(\psi^2_i-\psi^1_i\right)\beta S_i(t)}{A^1_i}\sum_{j=1}^{n}\lambda_{ij}\left(\frac{I^a_j(t)+(1-\eta_j) I^s_j(t)}{N_j}\right)\right)\right),\]
\[\theta^*_i=\max \left(0,\min \left( 1,\frac{\psi^2_iE_i(t)}{A^2_i}\right)\right),\,\,\text{and}\,\,\vartheta^*_i=\max \left(0,\min \left( 1,\frac{\psi^3_i I^a_i(t)+\psi^4_i I^s_i(t)}{A^3_i}\right)\right).\]
\end{theorem}

\begin{proof}
	From equation \eqref{hami}, we have the Hamiltonian at time $t$:
	\[H_i(t)=E_i(t)+I^a_i(t)+I^s(t)+\frac{A^1_i}{2}\left(\omega_i(t)\right)^2+\frac{A^2_i}{2}\left(\theta_i(t)\right)^2+\frac{A^3_i}{2}\left(\vartheta_i(t)\right)^2+\sum_{k=1}^{5}\psi^k_i(t)f_k\left(S_i, E_i, I^a_i, I^s_i,R_i\right)\]
	where  \begin{equation}
	\begin{aligned}
	f_1\left(S_i, E_i, I^a_i, I^s_i,R_i\right)=&\Gamma_i-\beta(1-\omega_i(t)) S_i(t) \sum_{j=1}^{n}\lambda_{ij}\left(\frac{I^a_j(t)+(1-\eta_j) I^s_j(t)}{N_j}\right)-\mu_i S_i,\\
	f_2\left(S_i, E_i, I^a_i, I^s_i,R_i\right)=&\beta S_i(t)(1-\omega_i(t)) \sum_{j=1}^{n}\lambda_{ij}\left(\frac{I^a_j(t)+(1-\eta_j) I^s_j(t)}{N_j}\right)-(\alpha_i+\mu_i)E_i(t)-\theta_i(t)E_i(t),\\
	f_3\left(S_i, E_i, I^a_i, I^s_i,R_i\right)=&(1-\sigma)\alpha_iE_i(t)-\left(\gamma^a_i+\delta^a_i+\mu_i\right)I^a_i(t)-\vartheta(t)I^a_i(t),\\
	f_4\left(S_i, E_i, I^a_i, I^s_i,R_i\right)=&\sigma\alpha_iE_i(t)-\left(\gamma^s_i+\delta^s_i+\mu_i\right)I^s_i(t)-\vartheta(t)I^s_i(t),\\
	f_5\left(S_i, E_i, I^a_i, I^s_i,R_i\right)=&\gamma^a_iI^a_i(t)+\gamma^s_iI^s_i(t)-\mu_iR_i(t),
	\end{aligned}
	\end{equation}	
	for time $t\in[0,T],$ adjoint equations and transversality conditions can be obtained by using  Pontryagin's maximum principle, such that
	
	\begin{equation}
	\begin{aligned}
	(\psi^1_i)^\prime=-\frac{dH_i}{dS_i}=&\psi^1_i\mu_i+(\psi^1_i-\psi^2_i)\beta\left(1-\omega_i(t)\right)\sum_{j=1}^{n}\lambda_{ij}\left(\frac{I^a_j(t)+(1-\eta_j) I^s_j(t)}{N_j}\right),\\
	(\psi^2_i)^\prime=-\frac{dH_i}{dE_i}=&\left(\alpha_i+\mu_i+\theta_i(t)\right)\psi^2_i-(1-\sigma)\alpha_i\psi^3_i-\sigma\alpha_i\psi^4_i,\\
	(\psi^3_i)^\prime=-\frac{dH_i}{dI^a_i}=&(\psi^1_i-\psi^2_i)\beta(1-\omega_i(t))S_i(t)\sum_{j=1}^{n}\left(\frac{\lambda_{ij}}{N_j}\right)+\left(\gamma^a_i+\delta^a_i+\mu_i+\vartheta(t)\right)\psi^3_i+\gamma^a_i\psi^5_i,\\
	(\psi^4_i)^\prime=-\frac{dH_i}{dI^s_i}=&(\psi^1_i-\psi^2_i)\beta(1-\omega_i(t))S_i(t)(1-\eta_j)\sum_{j=1}^{n}\left(\frac{\lambda_{ij}}{N_j}\right)+\left(\gamma^s_i+\delta^s_i+\mu_i+\vartheta(t)\right)\psi^4_i+\gamma^s_i\psi^5_i,\\
	(\psi^5_i)^\prime=-\frac{dH_i}{dR_i}=&\mu_i\psi^5_i,
	\end{aligned}
	\end{equation}
	with the tranversality conditions at time $T$  $\psi^1_i(T)=0,\psi^2_i(T)=1,\psi^3_i(T)=1,\psi^4_i(T)=1,\psi^5_i(T)=0.$ For $t\in[0,T]$, the optimal controls  $\omega^*_i,\theta^*_i$ and $\vartheta^*_i$ can be obtained by solving the following optimality  conditions 
	\[\frac{dH_i}{\omega_i}=0,\frac{dH_i}{\theta_i}=0,\, \text{and}\, \frac{dH_i}{\vartheta_i}=0.\]
	Therefore, we obtain 
	\begin{equation}\label{opt_con}
	\begin{aligned}
	A^1_i\omega_i(t)+\left(\psi^1_i-\psi^2_i\right)\beta S_i(t) \sum_{j=1}^{n}\lambda_{ij}\left(\frac{I^a_j(t)+(1-\eta_j) I^s_j(t)}{N_j}\right)&=0,\\
	A^2_i\theta_i(t)-\psi^2_iE_i(t)&=0,\\
	A^3_i\vartheta_i(t)-\psi^3_i I^a_i(t)-\psi^4_i I^s_i(t)&=0.
	\end{aligned}
	\end{equation}
	Form system of equations \eqref{opt_con}, we obtain
	\[\omega_i(t)=\frac{\left(\psi^2_i-\psi^1_i\right)\beta S_i(t)}{A^1_i}\sum_{j=1}^{n}\lambda_{ij}\left(\frac{I^a_j(t)+(1-\eta_j) I^s_j(t)}{N_j}\right),\]
	\[\theta_i(t)=\frac{\psi^2_iE_i(t)}{A^2_i},\,\vartheta_i(t)=\frac{\psi^3_i I^a_i(t)+\psi^4_i I^s_i(t)}{A^3_i}.\]
	Using the bounds of the controls which is given by $C_{ab},$ we obtain the following optimal controls 
	\[\omega^*_i=\max \left(0,\min \left( 1,\frac{\left(\psi^2_i-\psi^1_i\right)\beta S_i(t)}{A^1_i}\sum_{j=1}^{n}\lambda_{ij}\left(\frac{I^a_j(t)+(1-\eta_j) I^s_j(t)}{N_j}\right)\right)\right),\]
	\[\theta^*_i=\max \left(0,\min \left( 1,\frac{\psi^2_iE_i(t)}{A^2_i}\right)\right),\,\,\vartheta^*_i=\max \left(0,\min \left( 1,\frac{\psi^3_i I^a_i(t)+\psi^4_i I^s_i(t)}{A^3_i}\right)\right)\]
\end{proof}

\section{Numerical Simulations}\label{num_simu}
Here, we illustrate the numerical results to our proposed multi-group  model \eqref{mod}.
\subsection{Comparison of contact pattern for four countries (India, Italy, Brazil and USA) at different locations}
On  August 20, 2020, the Ministry of Home Affairs (MHA), India, released the rules for activities allowed in Unlock $4.0$ \cite{unl4}. It states that lockdown will continue in the high-risk zones till $30^{th}$ September 2020, the same as the previous month (August). Outside the containment region, some events were allowed in a specific manner. The metro rail was permitted to be restarted in the categorized mode from September 7, 2020. The gathering of limited peoples was allowed; fifty people were allowed in marriage functions, twenty people were legalized in funereal/last rites ceremonies, and up to one hundred people were acceptable in religious, political, sports, and academic functions. Face masks wearing was mandatory in markets, workplaces, and during traveling. On September 30, 2020, the MHA delivered the guidelines for activities that were allowed in Unlock $5.0$ \cite{unl5} and announced that lockdown should continue in force stringently in the containment region still September 30, 2020 \cite{unl6}. States and Union Territories will be able to take some necessary actions from $15^{th}$ October to reopen the schools in a group manner. Cinema halls could be reopened from $10^{th}$ October 2020 with the half seating of their capacity, and Swimming pools were allowed to be used for training of sportsperson \cite{unl5}. On $27^{th}$ October, 2020 the MHA announced the instructions for Unlock $6.0$ \cite{unl5e,mha2710,mha2710a,unl6} and said the guidelines of Unlock 5.0 would continue to be employed in November too [13]. Out of containment zones, some states have allowed more activities and partial reopening of schools \cite{unl6a}. Thus, the Ministry of Home Affairs released separate guidelines to each Unlock every month. The movements of people vary according to the guidelines of every Unlock. Therefore, the contribution of social contacts in the disease transmissions depends on the Unlock.

First of all, we compare the social contact structure of people of four countries (India, Italy, Brazil, and USA). The data for the age distributions of people for these countries are collected from the Population Pyramid website \cite{pop_pyra} as depicted in Fig.  \ref{age_dis}. In these age distributions, after selecting for the different age groups, the age of each individual is distributed uniformly from the associated range.   The social contact matrices for different locations (home, work, school, others) are found in Prem et. al \cite{ pre2017} that are obtained from surveys and Bayesian imputation and represented by Fig. \ref{soci_con}. \\
This comparison aims to highlight the differences between their contact patterns and emphasize their effects on the transmission of infectious diseases. Fig. \ref{age_dis} shows the percentage of the population (separated by gender) in the five-year age group terminating at the age of eighty. The Taj Mahal dome shape of age distributions for India and Brazil is typical of a demographic transition. The narrower base of Italian populations is characteristic of aging people at or near sub-replacement fertility. The rectangle shape of the age distribution of the USA population represents the stationary distribution.   \\
	The second row of Fig. \ref{soci_con} shows the contact between different age groups in the household setting, represented by matrices $\lambda^h_{ij}$ where darker magenta colored squares indicate a more significant number of contacts. As discussed in \cite{pre2017}, the structures common to India, Italy, and Brazil are the diagonal dominance, reflecting communication between same age groups (i.e., siblings and partners) and the prominent off diagonals, separated by the mean inter-generation gap, reflecting connections between different age groups (i.e., children and parents). However, we have observed some other types of behavior in contact patterns at the home location in the case of USA. The principal difference in the contact matrix for India is the presence of a third dominant diagonal, again separated by the mean inter-generation gap, reflecting the incidence of three-generation households. It quantifies the significant contact between children and grandparents and the possibility of substantial transmission of contagion from third to the first generation. Such connections are smaller in Brazil and negligible in Italy.  \\
The third row of Fig. \ref{soci_con} shows the contact patterns $\lambda^s_{ij}$ between age groups at the school location in which we see that there is main diagonal dominant till the school ages (strongly assortative, with primary contacts within the school-going children), reflecting contact within school-age groups, and other smaller contacts reflecting the connection between student and teacher. As observed, this structure is typical in all four countries.   The fourth row of Fig. \ref{soci_con} shows the social contact matrices $\lambda^w_{ij}$ between age groups in the workplace. In disparity to home and school location, the contact pattern at the workplace is more homogeneous across age groups in all four countries, indicating that the workplace contributes to the transmission of contagion between age groups that are, otherwise, largely separated from each other in the household. The boundaries of these age groups are more prominent in India and USA than in Italy and Brazil.  The last row of Fig. \ref{soci_con} shows the matrices $\lambda^o_{ij}$  for contact patterns at other locations. It shows that the contacts are strongly assortative (forcefully contact within age groups) for India, Brazil, and USA,  reflecting the preferential social connection within age groups in this sphere, but otherwise do not show systematic patterns.  However, the contact pattern matrix for Italy reflects the non-preferential social connection. \\
In summary, in the case of India, the home location provides the main channel to disease transmission between three generations, and the workplace provides the main channel to (largely homogeneous) disease transmission between working-age groups. The school is the main channel of infection within children and between children and adult teachers to a smaller extent. Other spheres of contact, due to the assortative mixing, contribute to transmission within age groups.

\begin{figure}[H]
	\centering
	\includegraphics[width=4cm,height=7cm]{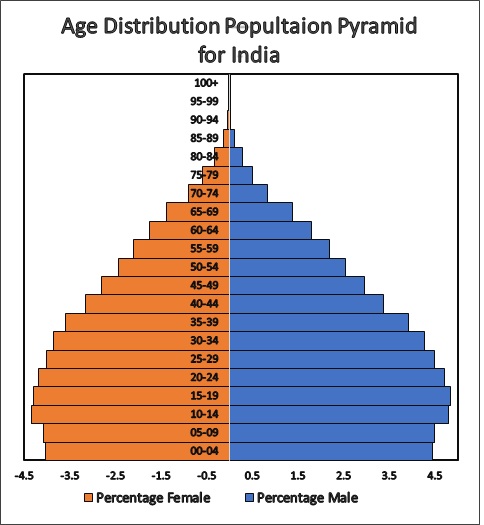}
	\includegraphics[width=4cm,height=7cm]{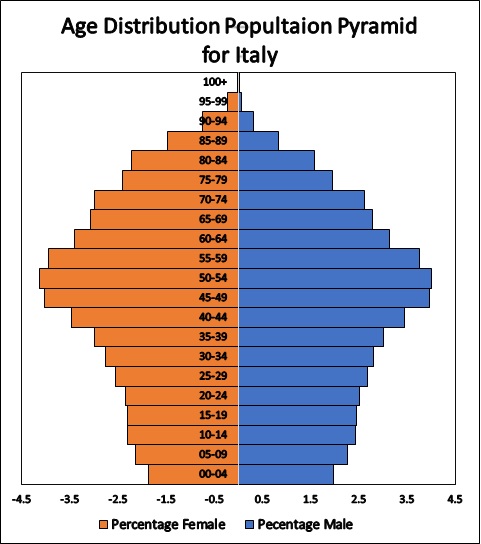}
	\includegraphics[width=4cm,height=7cm]{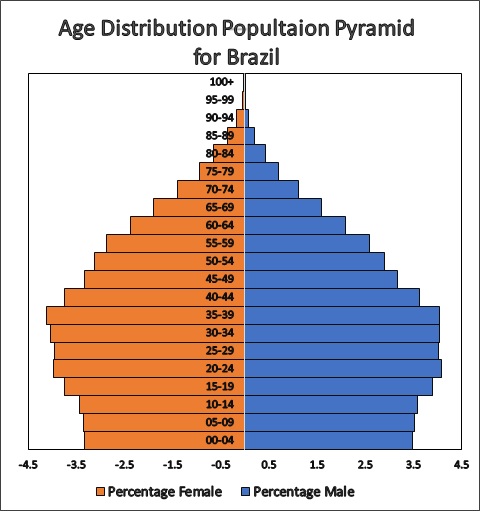}
	\includegraphics[width=4cm,height=7cm]{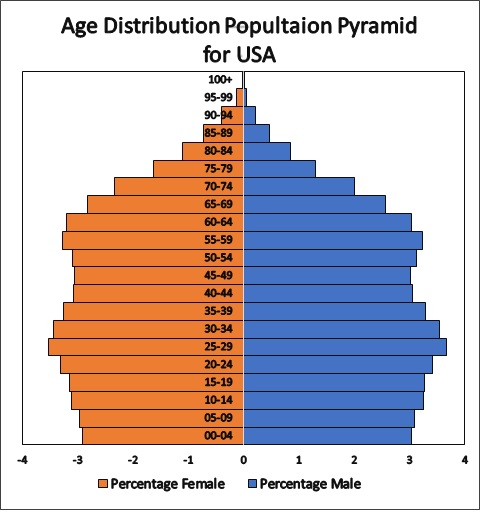}
	\caption{Age distribution population pyramids for India, Italy, Brazil and USA.}\label{age_dis}
\end{figure}

\begin{figure}[H]
	\centering
	\includegraphics[width=4cm,height=4cm]{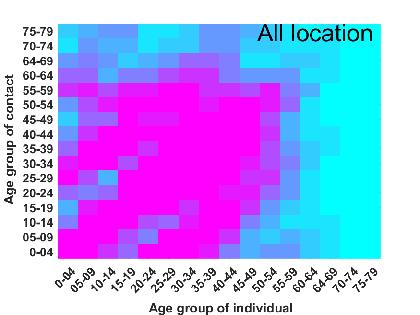}
	\includegraphics[width=4cm,height=4cm]{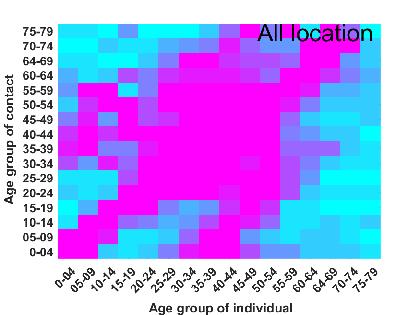}
	\includegraphics[width=4cm,height=4cm]{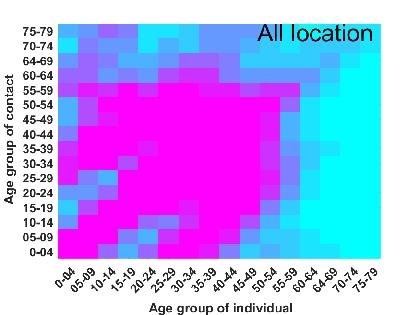}
	\includegraphics[width=4cm,height=4cm]{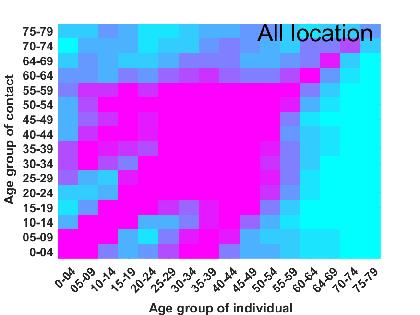}\\
	
	\includegraphics[width=4cm,height=4cm]{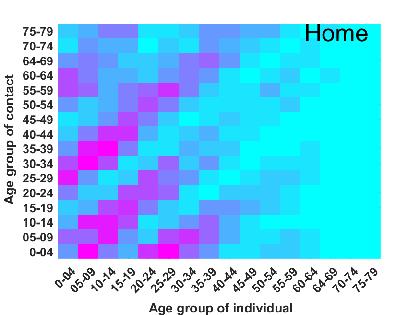}
	\includegraphics[width=4cm,height=4cm]{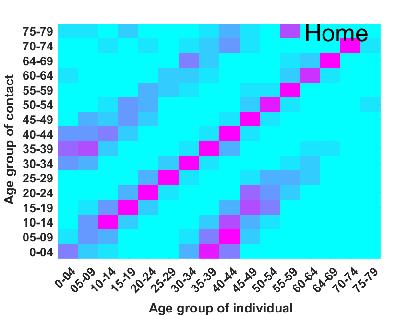}
	\includegraphics[width=4cm,height=4cm]{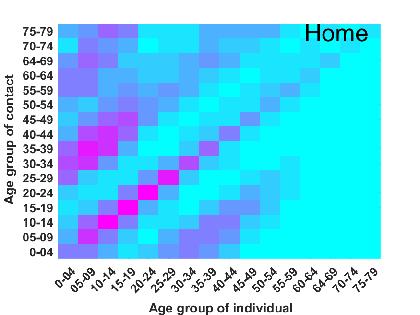}
	\includegraphics[width=4cm,height=4cm]{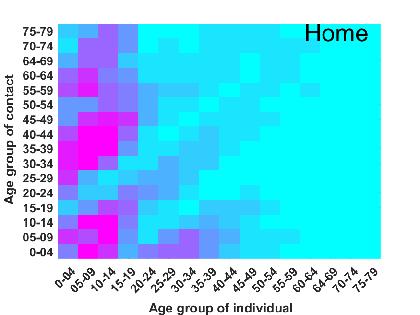}\\
	
	\includegraphics[width=4cm,height=4cm]{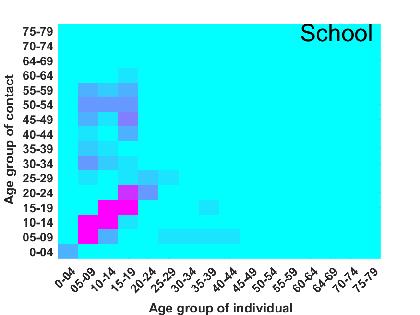}
	\includegraphics[width=4cm,height=4cm]{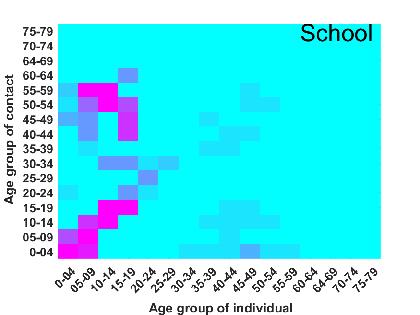}
	\includegraphics[width=4cm,height=4cm]{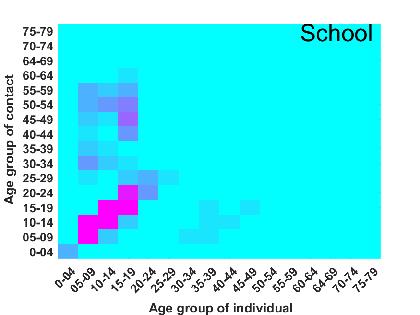}
	\includegraphics[width=4cm,height=4cm]{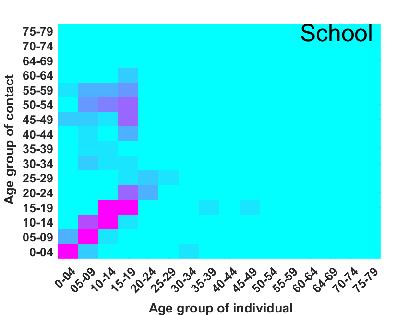}\\

	\includegraphics[width=4cm,height=4cm]{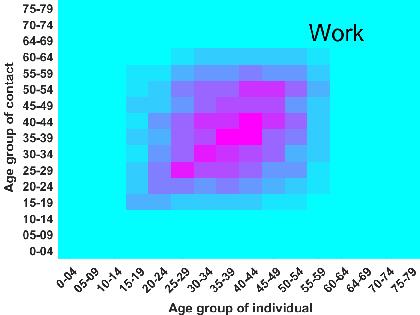}
	\includegraphics[width=4cm,height=4cm]{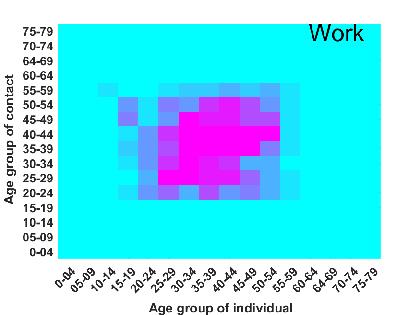}
	\includegraphics[width=4cm,height=4cm]{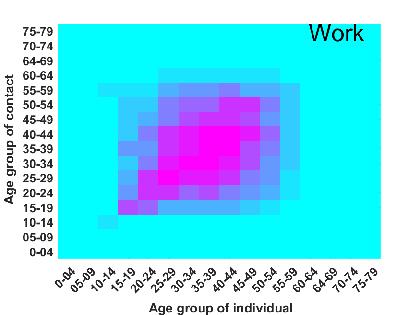}
	\includegraphics[width=4cm,height=4cm]{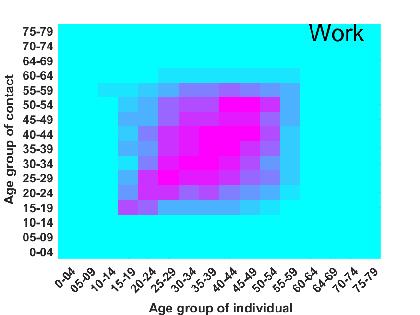}\\

	\includegraphics[width=4cm,height=4cm]{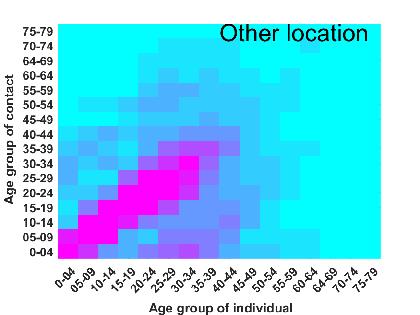}
	\includegraphics[width=4cm,height=4cm]{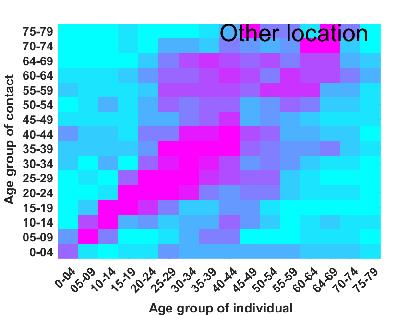}
	\includegraphics[width=4cm,height=4cm]{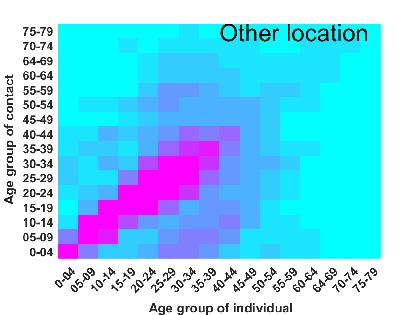}
	\includegraphics[width=4cm,height=4cm]{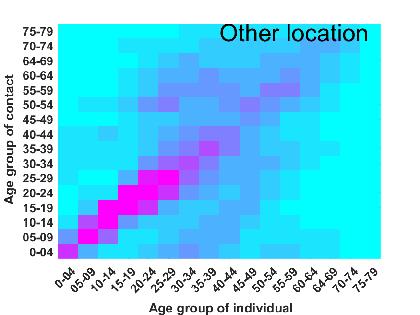}\\

	\caption{The social contact matrices of various locations (home, school, workplaces, and other locations) where the first column represents the matrices for India, the second column represents the matrices for Italy, the third column shows the matrices for Brazil and the last column  for USA.}\label{soci_con}
\end{figure}

Social mixing patterns vary according to locations, including homes, workplaces, schools, and other locations (markets, worship places, etc.). Therefore, we used the method set out in Prem et al. \cite{pre2017} which takes accounts for these differences and obtain the location-specific social contact matrices for various countries. In a standard setting, the average number of contacts made at all of these locations contributes to the overall mixing pattern in the population, i.e., the social contact matrices of the first row of Fig. \ref{soci_con} can be used for standard setting (before the outbreak). In an outbreak duration, different control measures (non-pharmaceutical interventions) are applied at reduced social mixing in different situations to lower the overall infection spreading in the community.   We used the weight coefficients between $0$ and $1$ to incorporate each intervention scenario from these building block matrices to simulate the effects of various control measures applied at reduced social mixing.

\subsection{Model fitting and parameter estimation: A case study of COVID-19 in India}
Since, from September 1, 2020, to December 31,  2020, the nationwide lockdown was lifted in a monthly phased manner. Every month, the strictness of lockdown was condensed, and the government of India permitted some more activities for people. On  August 20, 2020, the Ministry of Home Affairs (MHA), India, released the rules for activities allowed in Unlock $4.0$ \cite{unl4}. It states that lockdown will continue in the high-risk zones till $30^{th}$ September 2020, the same as the previous month (August). Outside the containment region, some events were allowed in a specific manner. The metro rail was permitted to be restarted in the categorized mode from September 7, 2020. The gathering of limited peoples was allowed; fifty people were allowed in marriage functions, twenty people were legalized in funereal/last rites ceremonies, and up to one hundred people were acceptable in religious, political, sports, and academic functions. Face masks wearing was mandatory in markets, workplaces, and during traveling. On September 30, 2020, the MHA delivered the guidelines for activities that were allowed in Unlock $5.0$ \cite{unl5} and announced that lockdown should continue in force stringently in the containment region still September 30, 2020 \cite{unl6}. States and Union Territories will be able to take some necessary actions from $15^{th}$ October to reopen the schools in a group manner. Many activities were accepted in a particularly phased manner \cite{unl5}. On $27_{th}$ October 2020, the MHA announced the instructions for Unlock $6.0$ \cite{unl5e,mha2710,mha2710a,unl6} and said, that the guidelines of Unlock 5.0 would continue to be employed in November too. Out of containment zones, some of the states have allowed more activities and partial reopening of schools \cite{unl6a}. Thus, the MHA released separate guidelines to each Unlock every month. \\
The movements of people vary according to the guidelines of every Unlock. Therefore, the contribution of social contacts in the disease transmissions depends on the Unlock. Hence, the social contact pattern was changing continuously, and the transmission rate differed accordingly.  Based on these observations, we consider the following four lockdown lifting phases: (i) September 2020 (ii) October 2020 (iii) November 2020 (iv) December 2020 and estimate the value of the transmission rate of COVID-19 infection for each phase. Therefore, we choose the empirical data of COVID1-19  for India for the period from September 1, 2020, to December 31,  2020 \cite{data}, to estimate the transmission rate of COVID-19 for the considered model. Our model fits the cumulative cases data for the period mentioned earlier, and the value of transmission rate $\beta$ has been estimated. It is noteworthy that the strictness of the intervention strategies can essentially change the disease's transmission dynamics.

Due to  the short study duration of COVID-19 epidemics compared to the human lifespan, we ignore the new recruitment rate in susceptible individuals and natural death rates of all individuals (demographics) i.e., we consider $\Gamma_i=0$ and $\mu_i=0$ for all $i=1,2,\cdots,16.$. We assume that the $40\%$ individuals of age groups $0-19$ years and the $80\%$ individuals of age groups $20-79$ years, are showing symptoms after the incubation period \cite{wang2021}. So, we set $\sigma=0.4$ for  $i=1,2,3,4$ and $\sigma=0.8$ for  $i=5,6,...,16.$ In terms of parametrization for incubation period of COVID-19, a reliable estimate of $\alpha_i$ for COVID-19 is difficult to obtain from observed data of COVID-19 infection. Therefore, we assume $\alpha_i$ to be a constant among all age groups, i.e., $\alpha_i=\alpha$ for $i^{th}$ age group.   Thus, we set $\alpha_i=\alpha=0.196$ (i.e. $1/\alpha (\text{latent period})=5.1$ \cite{lau2020}). For similar reason, further, we aslo assume $\gamma^a_i$ and $\gamma^s_i$ to  be a constant $\gamma^a_i=\gamma^s_i=\gamma_i$  among all age groups, i.e., $\gamma_i =\gamma=\frac{1}{\text{Infectious period}}$ for $i^{th}$ age group.  Therefore, the following system of differential equations  COVID-19 dynamic.

\begin{equation}
\begin{aligned}
\frac{dS_i(t)}{dt}=&-\beta S_i(t) \sum_{j=1}^{n}\lambda_{ij}\left(\frac{I^a_j(t)+(1-\eta_j) I^s_j(t)}{N_j(t)}\right),\\
\frac{dE_i(t)}{dt}=&\beta S_i(t) \sum_{j=1}^{n}\lambda_{ij}\left(\frac{I^a_j(t)+(1-\eta_j) I^s_j(t)}{N_j(t)}\right)-\alpha E_i(t),\\
\frac{dI^a_i(t)}{dt}=&(1-\sigma_i)\alpha E_i(t)-\left(\gamma+\delta^a_i\right)I^a_i(t),\\
\frac{dI^s_i(t)}{dt}=&\sigma_i\alpha E_i(t)-\left(\gamma+\delta^s_i\right)I^s_i(t),\\
\frac{dR_i(t)}{dt}=&\gamma \left(I^a_i(t)+I^s_i(t)\right),
\end{aligned}
\end{equation}
where $i\in \{1, 2, \cdots, 16\}$ and $n=16$ represents the number of age groups for simulations purpose of COVID-19 for India and the social contact matrix $\lambda_{ij}$ is  parameterized as \[\lambda_{ij}=\alpha_h\lambda^h_{ij}+\alpha_s\lambda^s_{ij}+\alpha_w\lambda^w_{ij}+\alpha_o\lambda^o_{ij}.\]
Further, we have that the limited activities are allowed to people in September 2020. Thus, we assume that $30\%$ social contacts at workplaces and $20\%$ contacts at other places of the normal setting are possible and generate the synthetic social contact matrices for workplaces and other places using weight coefficients $\alpha_w=0.3$ and $\alpha_o=0.2$ for September 2020, respectively. In similar manner, we set weight coefficients $\alpha_w=0.5$ and $\alpha_o=0.4$ for October; $\alpha_w=0.6$ and $\alpha_o=0.65$ for November;  $\alpha_w=0.65$ and $\alpha_o=0.75$ for December, 2020  according to guidelines of lockdown lifting.  We collected the real data of COVID-19 for India from September 1, 2020, to December 31, 2020, from the website \cite{data}. Further, we perform our model fitting by using the in-built function \textit{lsqnonlin} in Matlab 2018 to minimize the sum of square function and use the following steps;\\  
\textbf{Algorithm for Least Square Method for Data Fitting:}\\
Step 1. We obtain the cumulative cases of model output at time $t$ by summation of cumulative cases of each age group at time $t$. It is given by   $$I_{cum}(t,\Theta)=\sum_{i=1}^{n}\int_0^t \alpha_iE_i(\tau)d\tau, \,\,\,\,\text{where} \,\theta=\{\beta\}, $$ where $I_{cum}(t,\Theta)$ represents the cumulative confirmed cases at time $t$ and $n$ be the number of age groups.  For the initial values of the population of each age group, we assume that the total initial population of a class could be divided into age-specific compartments according to the age distribution pyramid discussed in Fig. \ref{age_dis}. \\
Step 2. We compute the sum of squares of errors at each time point, given by
$$ SS_k(\Theta)=\sum_{j=1}^m\left(I_{cum}^a(t_j)-I_{cum}(t_j,\Theta)\right)^2,$$
where, $I_{cum}^a(t_j)$ is the actual data at $t_j^{th}$  day and $m$ is the number of data points. $k$ represents the number of iterations performed  such that $SS_k$ is the sum of squares of errors in the $k^{th}$ iteration.\\
Step 3. We compute the value of parameter $\Theta$ such that
$$SS=\min \left\{SS(\Theta)\right\}.$$
Thus, we obtain the value of the transmission rate $(\beta)$ for the best fitting of model output to actual data. The fitting of model output to real data of cumulative confirmed cases is shown in Fig. \ref{data_fit}.  The estimated values of the transmission rate $\beta$ are displayed in Fig. \ref{esti_2}. The basic reproduction number is also estimated according to lockdown lefting in India which has been shown in Fig. \ref{esti_R0}.


\begin{figure}[H]
	\centering
	\includegraphics[width=15cm,height=10cm]{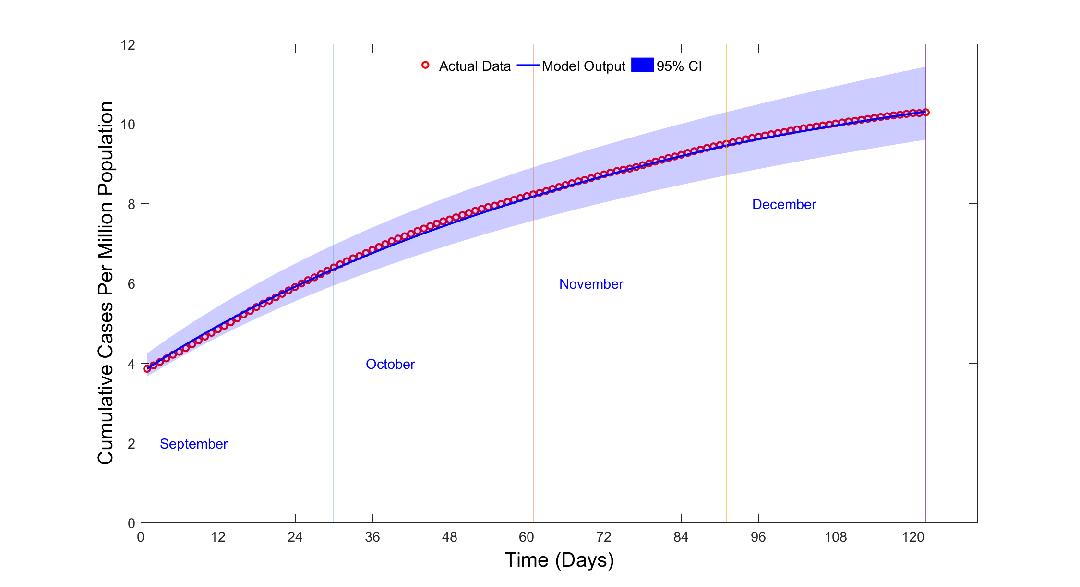}
	\caption{The best fitting of model output to cumulative confirmed cases. The month-wise fitting incorporates  different interventions scenarios applied by MHA in India.}\label{data_fit}
\end{figure}



\begin{figure}[H]
	\centering
	\includegraphics[width=15cm,height=8cm]{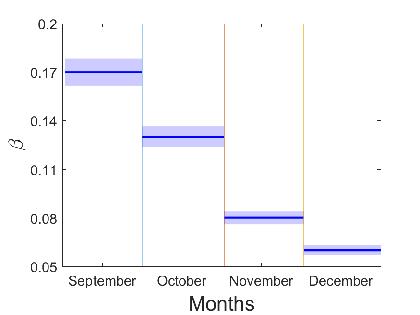}
	\caption{The figure illustrates the numerical values of model parameter $\beta$ for different age groups to the best fitting of model output according to previous months (September, October, November, and December) of 2020 (different-different scenarios of unlocking in India), where light blue shadowed area show the $95\%$ CI for estimated  $\beta$.}\label{esti_2}
\end{figure}

\begin{figure}[H]
	\centering\includegraphics[width=15cm,height=8cm]{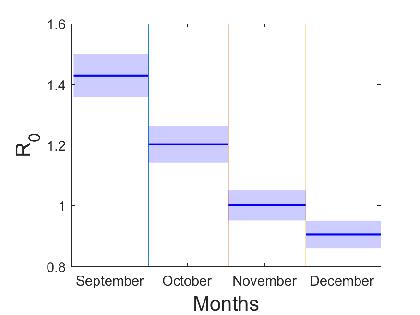}
		\caption{The figure illustrates the estimated basic reproduction number $R_0$ according to different-different scenarios of lockdown lifting in India, where light blue shadowed area show the $95\%$ CI for estimated  $R_0$.}\label{esti_R0}
\end{figure}
The basic reproduction number is constantly reducing, as indicated by the estimated values of $R_0$ in Fig. \ref{esti_R0}. It has also been demonstrated that the $R_0$ is more than one in September, October, and November of 2020, but less than one in December.
\subsection{Age profile of Asymptomatic and Symptomatic Infected Individuals}

\begin{figure}[H]
    \centering
\includegraphics[width=8cm,height=8cm]{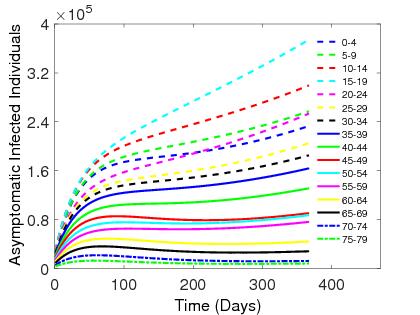}
\includegraphics[width=8cm,height=8cm]{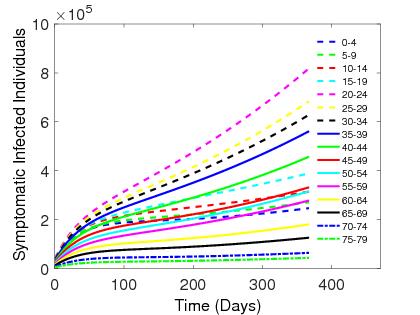}
    \caption{The figure depicts the long term dynamics of infected individuals of different age groups.}
    \label{age_prof}
\end{figure}
The age profile of infected individuals (asymptomatic and symptomatic ) are shown in Fig. \ref{age_prof}, in which we observe that asymptomatic infected individuals of age groups 05-24 years are gradually increasing. The possible reason behind this type of dynamics of infected individuals is the unawareness of children.     The asymptomatic infected individuals of age groups higher than 60 years increase and then decrease to stable levels due to their limited contacts (fewer contacts) to other persons. Moreover, the asymptomatic infected individuals of middle age groups (25-59 years) increase to their saturation levels. The individuals aged 25-59 years could not have important contacts due to their jobs, but they use protection measures like face masks, social distancing, etc.  Therefore, the dynamics of infected individuals of age groups 25-59 years is very different from those of 10-24 years. The same reasoning is also applicable for the dynamics of symptomatic infected individuals (please refer to the right panel of Fig. \ref{age_prof}).
\subsection{Impact of Awareness of Symptomatic Individuals}
	From the expression \eqref{R01} (which represents the basic reproduction number $(R_0)$ of the single group $SEI_aI_sR$ epidemic model), we noticed that $R_0$ would decrease by decreasing the values of $\eta$, i.e., the increase in the awareness in symptomatic individuals implies an increase in the avoiding unnecessary contacts. Consequently, the disease burden reduces to a significant level. Since the expression of the basic reproduction of system \eqref{mod} has a complicated form. Therefore, it is not easy to discuss the impact of awareness in Symptomatic individuals on basic reproduction number $R_0$ (expressed in equation \eqref{R0}) of age-structured epidemic model \eqref{mod} explicitly. Therefore, we observe the effects of awareness of symptomatic individuals by time series analysis of total infected individuals with various levels of understanding of people in Fig.  \ref{0alphaawar} when schools are closed, i.e., $\alpha_s=0$ and Fig. \ref{05alphaawar} when schools are partially opened, i.e., $\alpha_s>0$. In Fig. \ref{0vary_eta_1_4}, we observe the effects of symptomatic infected individuals of age groups $0-19$ years and notice that the number of total infected individuals is reduced by $17\%$ (approximately) at the end of the year due to increases in the awareness of symptomatic infected individuals of age groups $0-19$ years by $75\%$.  From Fig. \ref{0vary_eta_5_9}, we obtain that the number of total infected individuals is reduced by $35\%$ (approximately) at the end of the year due to increases in the awareness of symptomatic infected individuals of age groups $20-49$ years by $100\%$.  In Fig. \ref{0vary_eta_10_16}, we observe that the changes in the number of total infected individuals are significantly less (approximately $10\%$ reduction) at the end of the year due to increases in the awareness of symptomatic infected individuals of age groups $50-79$ years by $100\%$. Approximately $40\%$ reduction in the number of the total infected population is noticed due to a $100\%$ increment in awareness of symptomatic infected individuals of age groups $20-79$ years (please refer the Fig. \ref{0vary_eta_5_16}).


\begin{figure}[H]
    \centering
    	\subfloat[]{\includegraphics[width=8cm,height=8cm]{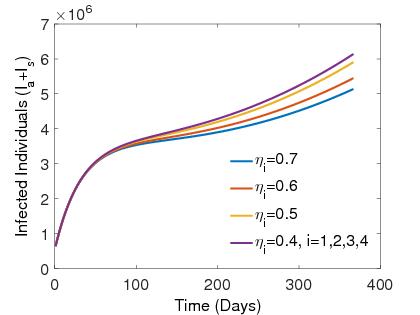}\label{0vary_eta_1_4}}\qquad
	\subfloat[]{\includegraphics[width=8cm,height=8cm]{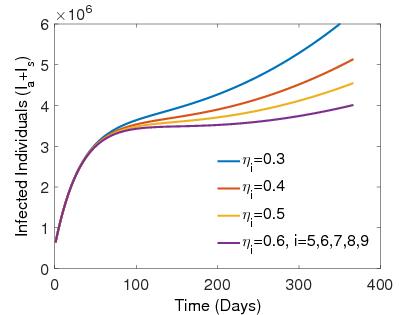}\label{0vary_eta_5_9}}\\
	\subfloat[]{\includegraphics[width=8cm,height=8cm]{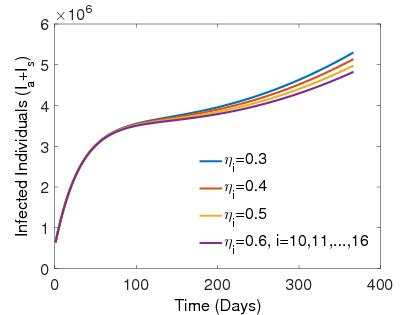}\label{0vary_eta_10_16}}\qquad
	\subfloat[]{\includegraphics[width=8cm,height=8cm]{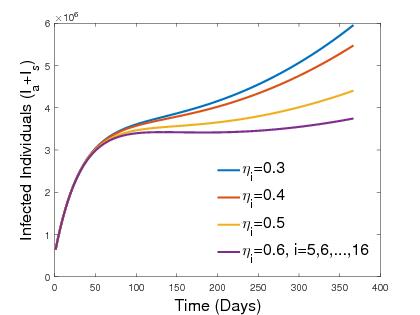}\label{0vary_eta_5_16}}
    \caption{The figure illustrates the impacts of awareness of symptomatic infected individuals on the long term dynamics of infected individuals when schools are closed $(\alpha_s=0).$}
    \label{0alphaawar}
\end{figure}

\begin{figure}[H]
    \centering
     	\subfloat[]{\includegraphics[width=8cm,height=8cm]{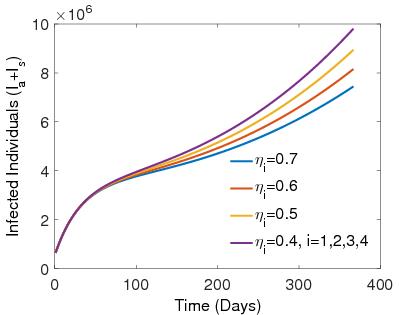}\label{05vary_eta_1_4}}\qquad
	\subfloat[]{\includegraphics[width=8cm,height=8cm]{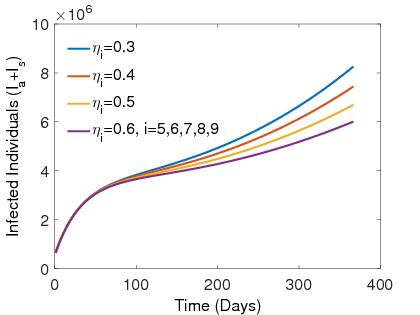}\label{05vary_eta_5_9}}\\
	\subfloat[]{\includegraphics[width=8cm,height=8cm]{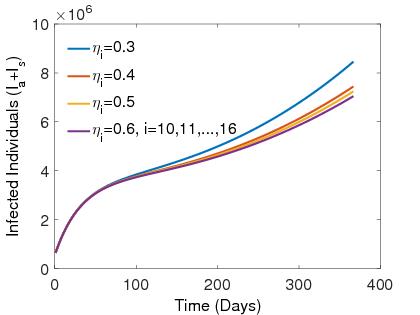}\label{05vary_eta_10_16}}\qquad
	\subfloat[]{\includegraphics[width=8cm,height=8cm]{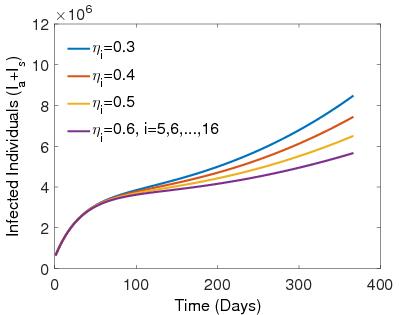}\label{05vary_eta_5_16}}
    \caption{The figure illustrates the impacts of awareness of symptomatic infected individuals on the long term dynamics of infected individuals when schools are partially  opened $(\alpha_s>0).$}
    \label{05alphaawar}
\end{figure}
From Fig. \ref{0vary_eta_1_4} and \ref{05vary_eta_1_4}, we have that the awareness of young people (with age $20-49$ yrs) is the most potent control measure to decrease the number of infected people when schools are partially opened or fully closed. Therefore, the awareness programs for people of age $20-49$ yrs should be exhibited by the authority to reduce the endemic level to a significant level when schools have opened partially or fully closed.
\subsection{Impact of Social mixing pattern }

\begin{figure}[H]
    \centering
   \subfloat[]{\includegraphics[width=8cm,height=8cm]{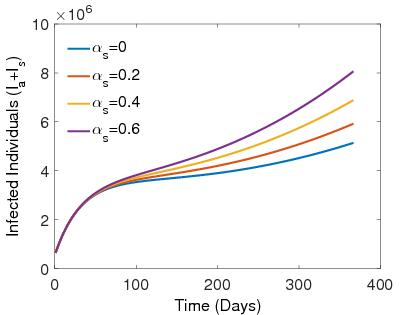}\label{vary_alpha_s}}
   \subfloat[]{\includegraphics[width=8cm,height=8cm]{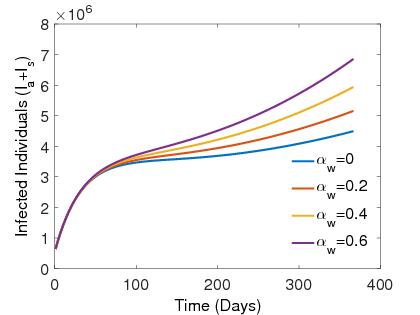}\label{vary_alpha_w}}\\
    \subfloat[]{\includegraphics[width=8cm,height=8cm]{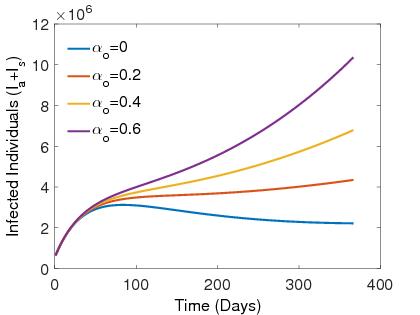}\label{vary_alpha_o}}
     \subfloat[]{\includegraphics[width=8cm,height=8cm]{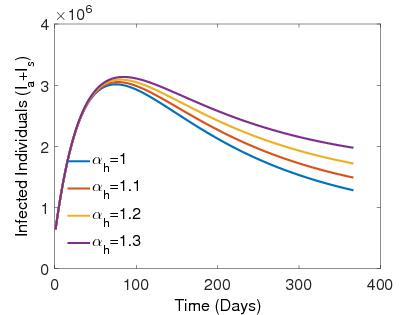}\label{vary_alpha_h}}
    \caption{The figure illustrates the impacts of social mixing pattern of people on the long term dynamics of infected individuals.}
    \label{socmix}
\end{figure}
In Fig. \ref{socmix}, we plotted long-term dynamics of total infected individuals for various social mixing scenarios. In the first three figures of  \ref{socmix} (Fig. \ref{vary_alpha_s}, \ref{vary_alpha_w} and \ref{vary_alpha_o}), we observe that the total number of infected individuals for four different levels of social contacts of a particular location and the weight coefficient of other three contacts matrices are fixed at a positive value.    In Fig. \ref{vary_alpha_s}, we observe that the effects of various levels of school reopening during the lockdown lifting when other three social contacts matrices are contributing to disease transmission, i.e., $\alpha_h>0,\alpha_w>0$ and $\alpha_o>0$ are taken in the figure.  Here, we notice that if $60\%$ schools are opened, the number of infected individuals is raised by approximately  $59\%$  from that of obtained in case of totally schools closure at the end of the year. In Fig. \ref{vary_alpha_w}, we noticed that about $53\%$ increment in the number of infected individuals is reported at the end of the year when  $60\%$  contacts of normal setting at working places are considered. In the end of the year, the huge amount of increment is reported in total number of infected individuals, when we increase the contact of other location up to $60\%.$ It can be seen in the Fig. \ref{vary_alpha_o}. From Fig. \ref{vary_alpha_h}, we have that as the value of $\alpha_h$ increase upto $30\%$ from $\alpha_h=1$, the total number of infected individuals also increases by $54\%.$ When a strict lockdown is implemented, people are more likely to stay at home for longer periods of time than in a normal scenario. As a result, the value of the weight coefficient for the home location $(\alpha_h)$ rises. Hence, the number of infected cases may rise. Thus , we can deduce that a long-term rigorous lockdown may be dangerous to those who stay home.
From Fig. \ref{socmix}, we conclude the changes in contacts at other location (shopping centers, cinema halls, restaurants, etc. ) also make significant changes in total number of infected population.  Hence, the number of infected people is controlled to a significant level by reducing the social contacts at other places (market, shopping malls, parks, etc.).

\section{Discussion}\label{con}
COVID-19, a contact-contagious infectious disease, is assumed to spread through a population via direct contacts between peoples \cite{gog2020,cha2020}. Outbreak control measures aimed at reducing the amount of mixing in the population can reduce the final size of the epidemic. Mathematical models help us to understand that how COVID-19  could spread across the population and inform control measures that might mitigate future transmission\cite{lan2020,chan2020}. 
The transmission dynamics of an infectious disease are most sensitive to the social contact patterns in a population of a particular community and to analyze the precautions people use to reduce the transmission of the disease. The social contact pattern depends on the age distribution of the specific community via different location such as  work, school and recreation etc. Therefore, knowing the age-specific prevalence and incidence of the infectious disease is essential for modeling the future burden of the disease and the effectiveness of interventions such as vaccination.
We simulated the shape of the ongoing outbreak of COVID-19 in India using an age-structured multi-group SEIR model. We accounted for various types of heterogeneities in contact networks in our model system \eqref{mod} to determine the effect of location-specific physical distancing measures such as extended school closures and interventions in workplaces and other gathering places on the disease progression. We simulated the COVID-19 outbreak and modeled the interventions by scaling down the appropriate component of the contact mixing matrices for India.  As individuals mixing patterns are non-random, they influence the disease transmission dynamics \cite{pre2020, dav2020}, models that evaluate the effectiveness of physical distancing interventions, such as school closure, need to account for social structures and heterogeneities in the mixing of individuals \cite{ wal2006, rea2003, edm1997}.

In our proposed model, we combined the changes to age-specific and location-specific social mixing patterns to estimate the effects of location-specific physical distancing interventions in controlling the transmission of the outbreak. The control measures put in place to reduce contacts at schools, workplaces, and other gathering places help us prevent the epidemic by providing the healthcare system sufficient time and opportunity to respond. Subsequently, if these restrictions are lifted early, there are still adequate susceptible people to keep the $R_0>1,$ once contacts increase and the number of infections grows. However, in realistic scenarios, the interventions are lifted slowly in a specific manner, partially, as an effort to avoid a sharp growth in infection due to logistical and practical reasons. Therefore, we simulated variously lifting the interventions such as school reopening in phase manners, market reopening, and many daily activities under the particular guidelines. There is little evidence of the effects of different physical distancing measures on containing the outbreak, is known about the behavioral changes of individuals over time during the outbreak. Therefore, we supposed the effect that certain types of physical distancing have on age-specific and location-specific contact rates to model the impact of the physical distancing measures executed in India.

To mathematical analysis, we have shown the existence of solutions for the proposed model. The nonnegativity and boundedness of solutions also have been discussed to demonstrate the feasibility in the context of population dynamics. The basic reproduction number $(R_0)$ has been computed using the next-generation matrix approach. We have studied the global dynamics of the system \eqref{mod} in terms of the basic reproduction number to check the behavior of the system around its equilibria. We found that the disease eliminates from all age groups when $R_0<1$ and the disease remains in the population when $R_0>1$. Hence, we obtained that the basic reproduction number plays the role of a threshold quantity to disease elimination from all age groups, and $R_0<1$ is a sufficient condition to eliminate COVID-19. The optimal control problem formulation has also been discussed by considering three non-pharmaceutical controls, namely: the first is the effort to awareness campaigns program regarding taking care of their health to protect from COVID-19, and second is the effort to encourage the exposed individuals to join the quarantine centers to break the transmission of COVID-19. The last control is the effort to make aware the infected individuals stay home in isolation or join the institutional isolation ward. The existence of optimal control has been discussed for the proposed optimal control problem \eqref{opt_con}. The optimal control to minimize the total number of infected population also has been computed mathematically.

Our numerical results exhibited that control measures aimed at reducing social mixing in the population can effectively reduce the disease burden of the COVID-19 outbreaks. For different scenarios of lifting the control measures among individuals, the schools and other places reopening have had a significant impact on the progression of the epidemic. We present the age profile of the number of asymptomatic and symptomatic infected individuals for one year (365 days) in Fig. \ref{age_prof}. Our simulation result suggests that the rate of change in the number of asymptomatic and symptomatic infected individuals vary in different categories; the number of asymptomatic infected individuals of age groups $5-24$ year and the number of symptomatic infected individuals of age groups $20-44$ year fastly increase. The solution trajectories of other age groups are going to an almost saturation level. Further, we discussed the effects of awareness of symptomatic infected individuals (refer Fig.\ref{0alphaawar}) when all schools are fully closed, i.e., $\alpha_s=0.$ Here, we conclude that the awareness of symptomatic infected individuals with age groups $20-49$ year is beneficial to reduce the number of infected individuals. It should be noticed that the large number of people involved in the essential services during the COVID-19 pandemic belongs to age groups $20-49$ year. The contacts of these people were contributing to the transmission of COVID-19 on a large scale. Therefore, the disease burden may have been reduced significantly by their awareness. It becomes essential that healthcare workers and other service providers are trained for awareness differently. However, the reduction in contacts of school-age children also reduces infection in the community when all or some schools are partially opened, i.e., $0<\alpha_s<1.$   ( refer the Fig. \ref{05alphaawar}). The impact of different scenarios of social mixing patterns of people on the dynamics of infected individuals has been shown in Fig \ref{socmix}.

Issues about schools opening during a pandemic originate from the possibility of higher incidence rates from school-related and other social activities. Parents are worried about diseases in children and teens, and this may bring the danger into the home, exacerbating the risk of community transmission. As a result, the increased community transmission may be attributable to the reopening of the school. The increased community transmission risk after school reopening can really be completely blamed on school-related activities. Our model confirms that the increase in contact among people returning to work, as well as the increase in social and recreational time, are also key causal drivers. Furthermore, after schools reopened, there was an unavoidable increase in contacts in adults, children, and adolescents, highlighting the significance of tight NPIs to allow for a safe reopening, despite the minor increase in cases in children and youth. 

The main findings of this study are that the reduction of contacts at schools and other locations (market,  shopping malls, and other gathering places) is an effective tool to decrease the number of infected individuals to a significant level. The school closure is also suggested as an intervention to reduce the COVID-19 in Wuhan, China \cite{pre2020}. The school reopening with limited attendance, such as older students sitting exams or the partial reopening to younger year groups, do not appear to affect the community spreading of COVID-19 significantly. However, fully school reopening can contribute to statistically significant increases in the transmission rate for COVID-19 \cite{sta2021,yua2022}. Face mask-wearing also should be necessary for the public and aware to avoid unnecessary contact with other people. In the case of schools closure, the awareness of symptomatic infected individuals with age groups $20-49$ years is helpful to significantly decrease the number of infected individuals due to their involvement in the essential services during the COVID-19 pandemic. However, In the case of schools partially opening, children's awareness is also required to reduce the number of infected people and the awareness of symptomatic infected individuals with age groups $20-49$ years. Similar results were also discussed for United States in \cite{mon2021}.

In conclusion, the non-pharmaceutical interventions based on sustained physical distancing have a strong potential to reduce the size of the epidemic burden of COVID-19 and lead to a smaller number of overall cases. The pressure on the healthcare system could be reduced by lowering and flattening the epidemic curve. The reduction in the disease prevalence also helps to reduce the load on the healthcare system.  The premature and sudden lifting of control measures could lead to other peaks, and the endemic curve could be flattened by gradually relaxing the interventions. Therefore, gradually relaxing the interventions is mandatory to reduce the disease burden in the community.

\newpage


\begin{thebibliography}{99}
	
	
		\bibitem{gra2008}
		Grassly, Nicholas C., and Christophe Fraser. "Mathematical models of infectious disease transmission." Nature Reviews Microbiology 6.6 (2008): 477-487.  
		
		\bibitem{gog2020}
		Gog, Julia R. "How you can help with COVID-19 modelling." Nature Reviews Physics 2.6 (2020): 274-275.
		\bibitem{ves2020}
		Vespignani, Alessandro, Huaiyu Tian, Christopher Dye, James O. Lloyd-Smith, Rosalind M. Eggo, Munik Shrestha, Samuel V. Scarpino et al. "Modelling covid-19." Nature Reviews Physics 2, no. 6 (2020): 279-281.
		\bibitem{tho2020}
		Thompson, Robin N., T. Déirdre Hollingsworth, Valerie Isham, Daniel Arribas-Bel, Ben Ashby, Tom Britton, Peter Challenor et al. "Key questions for modelling COVID-19 exit strategies." Proceedings of the Royal Society B 287, no. 1932 (2020): 20201405.
		\bibitem{due2007}
		Duerr, H-P., M. Schwehm, C. C. Leary, S. J. De Vlas, and M. Eichner. "The impact of contact structure on infectious disease control: influenza and antiviral agents." Epidemiology \& Infection 135, no. 7 (2007): 1124-1132.
		
		\bibitem{pel2009}
		Pellis, Lorenzo, Neil M. Ferguson, and Christophe Fraser. "Threshold parameters for a model of epidemic spread among households and workplaces." Journal of the Royal Society Interface 6, no. 40 (2009): 979-987.
		
		\bibitem{bec1995}
		Becker, Niels G., and Klaus Dietz. "The effect of household distribution on transmission and control of highly infectious diseases." Mathematical Biosciences 127, no. 2 (1995): 207-219.
		
		
		\bibitem{kie2020}
		Kiesha, P., Y. Liu, and T. Russell. "The effect of controlling strategies to reduce social mixing on outcomes of the COVID-19 epidemic in Wuhan." China: A modeling study, LANCET PUBLIC Health, (20) (2020): 30073-6.
		
		
		\bibitem{bag2013}
		Baguelin, Marc, Stefan Flasche, Anton Camacho, Nikolaos Demiris, Elizabeth Miller, and W. John Edmunds. "Assessing optimal target populations for influenza vaccination programmes: an evidence synthesis and modelling study." PLoS Med 10, no. 10 (2013): e1001527.
		
		\bibitem{war2018}
		de Waroux, Olivier le Polain, et al. "Identifying human encounters that shape the transmission of Streptococcus pneumoniae and other acute respiratory infections." Epidemics 25 (2018): 72-79.
		
		\bibitem{laf2020}
		Latsuzbaia, Ardashel, Malte Herold, Jean-Paul Bertemes, and Joël Mossong. "Evolving social contact patterns during the COVID-19 crisis in Luxembourg." PloS one 15, no. 8 (2020): e0237128.
		
		
		\bibitem{kuc2014}
		A. J. Kucharski, K. O. Kwok, V. W. I. Wei, B. J. Cowling, J. M. Read, J. Lessler, D. A. Cummings, and S. Riley. The Contribution of Social Behaviour to the Transmission of Influenza A in a Human Population. PLoS Pathog, 10(6):e1004206, 2014. 
		
		
		\bibitem{bug2020}
		Bugalia, Sarita, Vijay Pal Bajiya, Jai Prakash Tripathi, Ming-Tao Li, and Gui-Quan Sun. Mathematical modeling of COVID-19 transmission: the roles of intervention strategies and lockdown. Mathematical Biosciences and Engineering 17, no. 5 (2020): 5961-5986.
		
		
		\bibitem{baj2020}
		Bajiya, Vijay Pal, Sarita Bugalia, and Jai Prakash Tripathi. "Mathematical modeling of COVID-19: impact of non-pharmaceutical interventions in India." Chaos: An Interdisciplinary Journal of Nonlinear Science 30, no. 11 (2020): 113143.
		
		
	
	\bibitem{mas2015}	
	 R. Mastrandrea, J. Fournet, and A. Barrat. Contact patterns in a high school: A comparison
		between data collected using wearable sensors, contact diaries and friendship surveys. PLoS ONE, 2015.
		
		
		
		\bibitem{sal2010}
		M. Salathé, M. Kazandjieva, J. W. Lee, P. Levis, M. W. Feldman, and J. H. Jones. A high resolution human contact network for infectious disease transmission. Proc Natl Acad Sci U
		S A, 107(51):22020–5, Dec 2010. 
		
		
		
		\bibitem{mos2008}
		Mossong, Joël, Niel Hens, Mark Jit, Philippe Beutels, Kari Auranen, Rafael Mikolajczyk, Marco Massari et al. "Social contacts and mixing patterns relevant to the spread of infectious diseases." PLoS Med 5, no. 3 (2008): e74.
		
		\bibitem{cau2008}
		Cauchemez, Simon, Alain-Jacques Valleron, Pierre-Yves Boelle, Antoine Flahault, and Neil M. Ferguson. "Estimating the impact of school closure on influenza transmission from Sentinel data." Nature 452, no. 7188 (2008): 750-754.
		
		\bibitem{eam2012}
		Eames, Ken TD, Natasha L. Tilston, Ellen Brooks-Pollock, and W. John Edmunds. "Measured dynamic social contact patterns explain the spread of H1N1v influenza." PLoS Comput Biol 8, no. 3 (2012): e1002425.
		
		\bibitem{moh2020}
		\url{https://www.mohfw.gov.in/covid_vaccination/vaccination/important-information.html}
		
		
		
		\bibitem{cao2020}
		Cao, Qing, Yi-Ching Chen, Chyi-Liang Chen, and Cheng-Hsun Chiu. "SARS-CoV-2 infection in children: Transmission dynamics and clinical characteristics." Journal of the Formosan Medical Association 119, no. 3 (2020): 670.
		
		\bibitem{pre2020}
		Prem, Kiesha, Yang Liu, Timothy W. Russell, Adam J. Kucharski, Rosalind M. Eggo, Nicholas Davies, Stefan Flasche et al. "The effect of control strategies to reduce social mixing on outcomes of the COVID-19 epidemic in Wuhan, China: a modelling study." The Lancet Public Health 5, no. 5 (2020): e261-e270.
		
		\bibitem{baj2021}
		Bajiya, Vijay Pal, Jai Prakash Tripathi, Vipul Kakkar, Jinshan Wang, and Guiquan Sun. Global Dynamics of a Multi-group SEIR Epidemic Model with Infection Age. Chinese Annals of Mathematics, Series B 42, no. 6 (2021): 833-860.
		
		

		
		\bibitem{jos2020}
		Wu, Joseph T., Kathy Leung, Mary Bushman, Nishant Kishore, Rene Niehus, Pablo M. de Salazar, Benjamin J. Cowling, Marc Lipsitch, and Gabriel M. Leung. "Estimating clinical severity of COVID-19 from the transmission dynamics in Wuhan, China." Nature medicine 26, no. 4 (2020): 506-510.
		
		
		\bibitem{kel2020}
		To, Kelvin Kai-Wang, Owen Tak-Yin Tsang, Wai-Shing Leung, Anthony Raymond Tam, Tak-Chiu Wu, David Christopher Lung, Cyril Chik-Yan Yip et al. "Temporal profiles of viral load in posterior oropharyngeal saliva samples and serum antibody responses during infection by SARS-CoV-2: an observational cohort study." The Lancet Infectious Diseases 20, no. 5 (2020): 565-574.
		
		\bibitem{zho2020}
		Zhou, Fei, Ting Yu, Ronghui Du, Guohui Fan, Ying Liu, Zhibo Liu, Jie Xiang et al. "Clinical course and risk factors for mortality of adult inpatients with COVID-19 in Wuhan, China: a retrospective cohort study." The lancet 395, no. 10229 (2020): 1054-1062.
		
		
		
		\bibitem{ver2020}
		Verity, Robert, Lucy C. Okell, Ilaria Dorigatti, Peter Winskill, Charles Whittaker, Natsuko Imai, Gina Cuomo-Dannenburg et al. "Estimates of the severity of coronavirus disease 2019: a model-based analysis." The Lancet infectious diseases 20, no. 6 (2020): 669-677.
		
		
		
		\bibitem{dav2020}
		Davies, Nicholas G., Petra Klepac, Yang Liu, Kiesha Prem, Mark Jit, and Rosalind M. Eggo. "Age-dependent effects in the transmission and control of COVID-19 epidemics." Nature medicine 26, no. 8 (2020): 1205-1211.
		
		\bibitem{zou2020}
		Zou, Lirong, Feng Ruan, Mingxing Huang, Lijun Liang, Huitao Huang, Zhongsi Hong, Jianxiang Yu et al. "SARS-CoV-2 viral load in upper respiratory specimens of infected patients." New England Journal of Medicine 382, no. 12 (2020): 1177-1179.
		
		
		\bibitem{jon2020}
		Jones, Terry C., Barbara Mühlemann, Talitha Veith, Guido Biele, Marta Zuchowski, Jörg Hoffmann, Angela Stein, Anke Edelmann, Victor Max Corman, and Christian Drosten. "An analysis of SARS-CoV-2 viral load by patient age." MedRxiv (2020).
		
		
		\bibitem{ayo2020}
Ayoub, Houssein H., Hiam Chemaitelly, Shaheen Seedat, Ghina R. Mumtaz, Monia Makhoul, and Laith J. Abu-Raddad. "Age could be driving variable SARS-CoV-2 epidemic trajectories worldwide." PLoS One 15, no. 8 (2020): e0237959.

\bibitem{chi2020}
Chikina, Maria, and Wesley Pegden. "Modeling strict age-targeted mitigation strategies for COVID-19." PloS one 15, no. 7 (2020): e0236237.


\bibitem{van2002}	
Van den Driessche, Pauline, and James Watmough. Reproduction numbers and sub-threshold endemic
equilibria for compartmental models of disease transmission. Mathematical biosciences 180.1-2 (2002):
29-48.	


\bibitem{fia2017}
Fiona M Guerra, Shelly Bolotin, Gillian Lim, Jane Heffernan, Shelley L Deeks, Ye Li, and Natasha
S Crowcroft. The basic reproduction number $(R_0)$ of measles: a systematic review. The Lancet
Infectious Diseases, 17(12):e420- e428, 2017.


\bibitem{ste2020}
Steven Sanche and Yen Ting Lin. High contagiousness and rapid spread of severe acute respiratory
syndrome coronavirus. EID Journal, 26(7), 2020.	


\bibitem{Liu2021}
Liu, Ying, and Joacim Rocklöv. The reproductive number of the Delta variant of SARS-CoV-2 is far higher compared to the ancestral SARS-CoV-2 virus. Journal of travel medicine (2021).
	
	
	\bibitem{shu2013}
	Shuai, Zhisheng, and Pauline van den Driessche. Global stability of infectious disease models using Lyapunov functions. SIAM Journal on Applied Mathematics 73.4 (2013): 1513-1532.
	
	
	\bibitem{ber1994}
	Berman, Abraham, and Robert J. Plemmons. Nonnegative matrices in the mathematical sciences. Society for Industrial and Applied Mathematics, 1994.
	
	\bibitem{las1976}
	LaSalle, J. P. The stability of dynamical systems, Regional Conf. Ser. Appl. Math. SIAM, Philadephia
	(1976).
	
	\bibitem{fre1994}
	Freedman, Herb I., Shigui Ruan, and Moxun Tang. Uniform persistence and 
flows near a closed positively invariant set. Journal of Dynamics and Differential Equations 6.4 (1994): 583-600.
	
	
	\bibitem{li1999}
	Li, M. Y., Graef, J. R., Wang, L., and Karsai, J. Global dynamics of a SEIR model with varying total
	population size. Mathematical biosciences, 160(2)(1999), 191-213.
	
	\bibitem{dip2009}
Diprima BWE, ,Elementary RC. Differential equations and boundary value problems. New York: John Wiley Sons; 2009.	

\bibitem{bir1989}	
Birkhoff, Garrett, and Gian-Carlo Rota. Ordinary differential equations. John Wiley Sons; 1989.

\bibitem{fle1975}
Fleming WH, Rishel RW. Deterministic and stochastic optimal control. New York, NY, USA: Springer; 1975.

\bibitem{pon1962}
Pontryagin LS, Boltyanskii  VG,
Gamkrelidze RV, and  Mishchenko EF. The Mathematical Theory of
Optimal Processes. New York, NY, USA: Wiley; 1962.


\bibitem{unl4}
Ministry of Home Affairs Unlock $4.0$ official guidelines on their official website (PDF). Ministry of Home Affairs.


\bibitem{unl5}
Ministry of Home Affairs Unlock $5.0$ official guidelines on their official website (PDF). Ministry of Home Affairs.


\bibitem{unl5e}
``MHA Order Dt, 27.10.2020 reg extension of guidelines of reopening upto 30.11.2020" (PDF). Ministry of Home Affairs, GoI. 27 October 2020. Retrieved 12 November 2020.



\bibitem{mha2710}
MHA DO Dt. 27.10.2020 to all Chief Secretaries reg extension of reopening upto 30.11.2020" (PDF). Ministry of Home Affairs, GoI. 27 October 2020. Retrieved 12 November2020.


\bibitem{mha2710a}
MHA DO Dt. 27.10.2020 to all Administrators reg extension of reopening upto 30.11.2020" (PDF). Ministry of Home Affairs, GoI. 27 October 2020. Retrieved 12 November 2020.


\bibitem{unl6}
Unlock 6.0 latest guidelines: What do they mean for schools, colleges, other educational institutions. DNA India. 29 October 2020. Retrieved 11 November 2020.

\bibitem{unl6a}
Unlock 6.0 guidelines: Which states have allowed more relaxations in Nov?. The Indian Express. 6 November 2020. Retrieved 11 November 2020.


\bibitem{pop_pyra}
\url{https://www.populationpyramid.net}



\bibitem{pre2017}
Prem, Kiesha, Alex R. Cook, and Mark Jit. Projecting social contact matrices in 152 countries using contact surveys and demographic data. PLoS computational biology 13.9 (2017): e1005697.


\bibitem{wang2021}
Wang, Xia, Hulin Wu, and Sanyi Tang. Assessing Age-Specific Vaccination Strategies and Post-Vaccination Reopening Policies for COVID-19 Control Using SEIR Modeling Approach. medRxiv (2021).

\bibitem{lau2020}
Lauer, Stephen A., Kyra H. Grantz, Qifang Bi, Forrest K. Jones, Qulu Zheng, Hannah R. Meredith, Andrew S. Azman, Nicholas G. Reich, and Justin Lessler. The incubation period of coronavirus disease 2019 (COVID-19) from publicly reported confirmed cases: estimation and application. Annals of internal medicine 172, no. 9 (2020): 577-582.


\bibitem{data}
\url{https://ourworldindata.org/covid-cases?country=IND~USA~GBR~CAN~DEU~FRA}


\bibitem{chan2020}
Chan, Jasper Fuk-Woo, Shuofeng Yuan, Kin-Hang Kok, Kelvin Kai-Wang To, Hin Chu, Jin Yang, Fanfan Xing et al. A familial cluster of pneumonia associated with the 2019 novel coronavirus indicating person-to-person transmission: a study of a family cluster. The lancet 395, no. 10223 (2020): 514-523.


\bibitem{lan2020}
Lancet, The. India under COVID-19 lockdown. Lancet (London, England) 395, no. 10233 (2020): 1315.

\bibitem{cha2020}
Chatterjee, Patralekha. Is India missing COVID-19 deaths?. The Lancet 396, no. 10252 (2020): 657.


\bibitem{wal2006}
Wallinga, Jacco, Peter Teunis, and Mirjam Kretzschmar. Using data on social contacts to estimate age-specific transmission parameters for respiratory-spread infectious agents. American journal of epidemiology 164, no. 10 (2006): 936-944.



\bibitem{rea2003}
Read, Jonathan M., and Matt J. Keeling. Disease evolution on networks: the role of contact structure. Proceedings of the Royal Society of London. Series B: Biological Sciences 270, no. 1516 (2003): 699-708.



\bibitem{edm1997}
Edmunds, W. John, C. J. O'callaghan, and D. J. Nokes. Who mixes with whom? A method to determine the contact patterns of adults that may lead to the spread of airborne infections. Proceedings of the Royal Society of London. Series B: Biological Sciences 264, no. 1384 (1997): 949-957.


\bibitem{sta2021}
Stage, Helena B., Joseph Shingleton, Sanmitra Ghosh, Francesca Scarabel, Lorenzo Pellis, and Thomas Finnie. Shut and re-open: the role of schools in the spread of COVID-19 in Europe. Philosophical Transactions of the Royal Society B 376, no. 1829 (2021): 20200277.


\bibitem{yua2022}
Pei Yuan, Aruffo, Elena, Nicholas Ogden, Evgenia Gatov, Sarah Collier, Qi Li, Iain Moyles, Nasri Bouchra, and Huaiping Zhu. School and community reopening during the COVID-19
pandemic: a mathematical modelling study. Royal Society Open Science, 9(2022): 211883.


\bibitem{mon2021}
Monod, Mélodie, Alexandra Blenkinsop, Xiaoyue Xi, Daniel Hebert, Sivan Bershan, Simon Tietze, Marc Baguelin et al. Age groups that sustain resurging COVID-19 epidemics in the United States. Science 371, no. 6536 (2021): eabe8372.

\end{thebibliography}
\end{document}